\documentclass[12pt]{article}%
\pdfoutput=1
\usepackage[utf8]{inputenc}%
\usepackage[T1]{fontenc}%
\usepackage{amsmath, amsthm, a4, latexsym, amssymb, mathrsfs, graphicx}%
\usepackage{natbib}%
\usepackage{bm}%
\usepackage{lineno}%
\usepackage[all]{xy}%
\usepackage{color}%
\usepackage{rotating}%
\usepackage{a4wide,fullpage}%
\usepackage{setspace}%
\usepackage{enumerate}%
\DeclareMathAlphabet{\mathpzc}{OT1}{pzc}{m}{it}%
 \newcommand{\Var}[1]{\ensuremath{\mathrm{Var}\!\left( #1\right) } }%
\newcommand{\Cov}[2]{\ensuremath{\mathrm{Cov}\!\left(#1, #2 \right) } }
\newcommand{\Covn}[3]{\ensuremath{\mathrm{Cov}_{#1}\!\left(#2, #3 \right) } }
\newcommand{\N}{\mathbb{N}}%
 \newcommand{\E}{\mathbb{E}}%
 \renewcommand{\P}{\mathbb{P}}%

 \newcommand{\hi}{\ensuremath{\hat{\imath}}}
\newcommand{\hj}{\ensuremath{\hat{\jmath} } }

 \renewcommand{\N}{{\mathbb N}}%

\newcommand{\ind}{\mathbb I}%
 \newcommand{\EE}[1]{\ensuremath{\E\!\left[ #1 \right] }}%
\newcommand{\bone}[1]{\ensuremath{ \boldsymbol{1}_{\left( #1 \right) } } }
\newtheorem{theorem}{Theorem}[figure]

\newtheorem{prop}[theorem]{Proposition}

\newtheoremstyle{thm}{1.5ex}{1.5ex}{\itshape\rmfamily}{}
{\bfseries\rmfamily}{}{2ex}{}

\newtheoremstyle{rem}{1.3ex}{1.3ex}{\rmfamily}{}
{\itshape\rmfamily}{}{1.5ex}{}
\theoremstyle{rem}
\refstepcounter{subsubsection}

\title{Statistical properties of the site-frequency spectrum associated with Lambda-coalescents}
\date{\today}
\author{Matthias Birkner\\Johannes-Gutenberg-Universität\\ Institut für Mathematik\\ 55099 Mainz, Germany\\\\ Jochen Blath und  Bjarki Eldon \\ TU Berlin\\ Institut für Mathematik\\  10623 Berlin, Germany  }
\begin{document}%
\maketitle

\clearpage
\pagebreak
\newpage

Running title:  Lambda-coalescents and inference\\
keywords:  Lambda-coalescents,  site-frequency spectrum, covariance, Atlantic cod, pseudo-likelihood\\\\

corresponding author:\\
Bjarki Eldon\\
TU Berlin, Institut für Mathematik \\
Stra{\ss}e des 17.~Juni 136 \\
 10623 Berlin, Germany \\
Email: eldon@math.tu-berlin.de\\
Phone: +49 303 1425 762\\
Fax: +49 +(0) 30 314 21695 \\

\clearpage
\pagebreak
\newpage

\begin{abstract}%

  Statistical properties of the site frequency spectrum associated
  with Lambda-coalescents are our objects of study. In particular, we
  derive recursions for the expected value, variance, and covariance
  of the spectrum, extending earlier results of Fu (1995) for the
  classical Kingman coalescent.  Estimating coalescent parameters
  introduced by certain Lambda-coalescents for datasets too large for
  full likelihood methods is our focus.  The recursions for the
  expected values we obtain can be used to find the parameter values
  which give the best fit to the observed frequency spectrum.  The
  expected values are also used to approximate the probability a
  (derived) mutation arises on a branch subtending a given number of
  leaves (DNA sequences), allowing us to apply a pseudo-likelihood
  inference to estimate coalescence parameters associated with certain
  subclasses of Lambda coalescents.  The properties of the
  pseudo-likelihood approach are investigated on simulated as well as
  real mtDNA datasets for the high fecundity Atlantic cod (\emph{Gadus
    morhua}).  Our results for two subclasses of Lambda coalescents
  show that one can distinguish these subclasses from the Kingman
  coalescent, as well as between the Lambda-subclasses, even for 
  moderate sample sizes.
  \end{abstract}

\clearpage
\pagebreak
\newpage





Large offspring number population models have recently been proposed
as appropriate models with which to investigate high fecundity natural
populations.  Some marine populations may belong to the class of high
fecundity populations, including Pacific oysters
\citep[\emph{Crassostrea gigas};][]{li98,B94,BCCHB02}, white sea bream
\citep[\emph{Diplodus sargus};][]{PL02}, and Atlantic cod
\citep[\emph{Gadus morhua};][]{A04}.  Oysters feature in \cite{Wi75}'s
elm and oyster model as an example of a high fecundity population.
Indeed, high fecundity populations are discussed at length by
\cite{Wi75} when comparing the benefits of sexual versus asexual
reproduction.  \cite{A88} compares genetic distances for mtDNA
variation for three vertebrate species, american eels (\emph{Anguilla
  rostrata}), hardhead catfish (\emph{Arius felis}), and red-winged
blackbirds (\emph{Agelaius phoeniceus}), and conclude that historical
effective population sizes may have been much lower than current
census size.  Low effective population size compared to census
population size observed for certain marine populations in particular
\citep[e.g.][]{HCW92}, and reviewed by \cite{HP11}, may be evidence of
high variance in offspring distribution.  Indeed, \cite{H05} observes
that low effective population size results from high variance in
reproductive success in a population with large census size.  High
fecundity may also be a way for certain marine organisms with broadcast
spawning to compensate for high mortality rate among juveniles, and thus
exhibiting Type III survivorship curves.

Multiple merger coalescent processes, so-called Lambda- and
Xi-coalescents, arise naturally from large offspring number models
\citep{S99,DK99,MS01,S03,EW06,SW08,BB09,HM11}.  A key assumption to obtain
multiple merger coalescent process from a large offspring number model
is that the number of offspring of an individual can be up to the
order of the population size with non-negligible probability.  In
contrast to the classical Kingman coalescent framework, which has
served as the standard null-model in population genetics since its
introduction \citep{K82}, Lambda-coalescents allow multiple instead of
only binary collisions of ancestral lineages, in which the group of
lineages coalescing in a given instance may include anywhere from two
to all of the current active lineages.  Selective sweeps
\citep{DS04,DS05,CR12}, or large-scale spatial extinction and
recolonisation events \citep{BEV10,BEV13}, can also give rise to
multiple merger coalescent processes.  In some cases, one naturally
encounters the even more general Xi-coalescents \citep{S00} allowing
{\em simultaneous} multiple mergers at a time, i.e., when distinct
groups of ancestral lineages coalesce at the same time to distinct
ancestors.  Models of severe bottlenecks \citep{BBMST09}, or diploidy
\citep{MS03} and recombination in conjunction with distinct loci can
give rise to Xi-coalescents.  By way of example, a Xi-coalescent
admitting up to quadrifold simultaneous mergers is obtained from a
diploid multi-locus model with a skewed (high variance) offspring
distribution \citep{BBE13}. \cite{MS01} provide a full mathematical
classification of coalescent limits (as population size tends to
infinity) of Cannings's exchangeable population models.

The statistical methodology to test for the presence and
particularities of the effects of multiple mergers, or to distinguish
between different underlying Lambda- and Xi-coalescents within a
family of models, is, however,  still largely missing. \cite{HP11} argue
that large offspring number models are the right framework in which to
study high fecundity marine populations.  Some large offspring number
models \citep[e.g.][]{S03,EW06} introduce new parameters
(\emph{coalescence} parameters) determining the size and intensity of
large offspring number events, i.e. when individuals have very many
offspring.  Estimating the coalescence parameters from genetic data
may give us a way to distinguish between different population models.
Indeed, \cite{HP11} call for the development of such inference
methods, in particular to infer parameters, assess the quality of the
estimators, and to construct hypothesis tests and confidence bounds.


On one hand are summary-statistics based estimators for the parameter
of certain classes of Lambda-coalescents \citep{E11} in the
infinite-sites model, such as the number of segregating sites, which
may not represent a sufficient statistic, but are computationally
efficient, and could be used to obtain prior information. On the
other hand, there are full-likelihood-based point estimators, derived
via relatively complex recursions - which are, in principle, amenable to
Monte-Carlo methods - taking all available information into account,
but whose practical treatment becomes computationally prohibitive even
for medium-sized datasets \citep{BB08,BBS11}. A small
empirical simulation study for the latter can be found in
\cite{S09}. For both types of estimators, an analysis of their
statistical properties in the Lambda-case is still largely missing.

A summary statistic that contains much of the essential (in
particular, topological) information of the underlying genealogical
tree of the data is the {\em site frequency spectrum (SFS)} (however,
one may want to consult e.g.\ \cite{M07} for theoretical
limitations). The spectrum is a key quantity when applying coalescent
theory in inference, see e.g.\ \cite{W07} for an overview and a
discussion of the relation of the SFS with various other (simpler)
summary statistics.  The statistical properties of the SFS under the Kingman coalescent  have been
investigated in several studies \citep{F95,GT98,A09,ZW08,ST11}.
Extending in particular the results of \cite{F95} to Lambda
coalescents is our aim.  Recursions for the expected value, the
variance and the covariance of the number of derived mutations in the
SFS, are derived for a general Lambda coalescent.

Substantial theoretical research has been undertaken in order to
identify the asymptotic behaviour of the frequency spectrum associated
with Lambda-coalescents, culminating in the recent a.s.\ convergence
results, as sample size tends to infinity, contained in
\cite{BBL12}. However, so far it is not known how large samples need
to be in order for the spectrum to be well-approximated by the
limiting asymptotic results of \cite{BBL12}, or, indeed, if the
approximation is equally good for all values of the relevant
coalescent parameter. This question is addressed using simulations.



Finally, as an illustration  we apply our methods to site-frequency spectra obtained for
Atlantic cod \citep{A00,SA03,A04} and fit Lambda coalescents to the
data. The problem of distinguishing between coalescent classes by
means of observing a realisation of the SFS is also addressed using
simulations.

\begin{center}
{\sc Setup and previously known results}%
\end{center}

\noindent {\bf Notation}

A formal definition of the site frequency spectrum (SFS) is our
starting point.  We assume the {\em infinitely-many-sites
  model}  \citep{K69,W75} of population genetics, where the genealogy of a sample of
size $n$ is described by an $n$-$\Lambda$-coalescent. Mutations 
 occur at a rate $\theta/2$ on the coalescent
branches. Suppose further that we can distinguish between mutant and
wild-type, e.g.\ with the help of an outgroup. Then, the SFS of an
$n$-sample is given by
$$
{\boldsymbol{\xi}}^{(n)} := \big(\xi_1^{(n)}, \dots, \xi^{(n)}_{n-1}\big),
$$
where the $\xi^{(n)}_i, i=1,\dots, n-1$ denote the number of sites at
which variants appear $i$-times in our sample.  The $\xi_i^{(n)}$ are
also referred to as the number of {\em derived} mutations in $i$
copies, when the ancestral state is known.  The
vector $\boldsymbol{\xi}^{(n)}$ is called the {\em unfolded}
site-frequency spectrum of the $n$-sample.

\newcommand{\hn}{\ensuremath{\lfloor n/2 \rfloor}}

In the case where the ancestral state (variant) is unknown, one often
considers $\eta^{(n)}_i = \xi_i^{(n)} $ for $i=1, \dots, \hn$, which
is the number of sites at which the less frequent variant (the minor
allele) is present in $i$ sequences, and $\hn := \bone{\textrm{$n$
    even}} n/2 + \bone{\textrm{$n$ odd}}(n-1)/2$, where $\bone{A} = 1$
if event $A$ is true, and zero otherwise.     The observed polymorphisms become `few and far between'
with corresponding small expected values as $i$ increases,   at least for the Atlantic
cod data we investigate.  Thus, for $i > \hat{\imath}$ for some
suitable $\hat{\imath}$, we lump them together into one class labelled
as $\hat{\imath}_+$.  Considering only  the first few classes and
lumping the rest also speeds up the computations.

Then, ${\boldsymbol{\eta}}^{(n)}:= \left(\eta_1^{(n)}, \dots,
  \eta^{(n)}_{\hn}\right)$ is called the {\em folded} site frequency
spectrum of an $n$-sample.


\newcommand{\one}{\mathbf{1}}

A {\em Lambda-coalescent} is a partition-valued (one can think of
enumerated DNA sequences and their ancestral relationships) exchangeable coalescent process determined by
a finite measure $\Lambda$ on $[0, 1]$ \citep{P99,S99,DK99}. For a comprehensive overview see e.g.\ \cite{B09}. If
there are currently $n$ blocks in the partition (ie.\ $n$ active
ancestral lineages), any particular set of $k$ blocks merges into one
at rate
\begin{equation} 
\label{eq:lambdarates}
\lambda_{n,k} = \int_{[0,1]} x^{k-2} (1-x)^{n-k} \Lambda(dx), \quad k=2,..,n.
\end{equation}
Certain special cases of a Lambda-coalescent include the following,
some of which we will use in our examples below:
\begin{itemize}
\item One recovers the Kingman coalescent for 
\begin{equation}
\label{eq:Kingman}
\Lambda(dx)= \delta_0(dx),
\end{equation} 
i.e.\ when the measure $\Lambda$ is concentrated at the point $0$.  
\item The Beta$(2-\alpha, \alpha)$-coalescent \citep{schweinsberg03} is a Lambda-coalescent characterised by the $\Lambda$-measure
\begin{equation}\label{eq:betameasure} 
\Lambda(dx) = \frac{\Gamma(2)}{\Gamma(2-\alpha)\Gamma(\alpha)} x^{1-\alpha} (1-x)^{\alpha-1} \,dx,
\end{equation}
with $\alpha \in (0,2)$; i.e.\ when the measure $\Lambda$ is associated
with the beta distribution with parameters $2-\alpha$ and $\alpha$.
The limiting case $\alpha=2$ (in the sense of weak convergence of
measures) corresponds to the Kingman coalescent.
\item \cite{EW06} considered purely atomic Lambda coalescents of the following types:
\begin{equation}
\label{eq:EWoneatom}
\Lambda(dx) =  \delta_\psi(dx), 
\end{equation}
and
\begin{equation}
\label{eq:EWtwoatom}
\Lambda(dx) =\frac{2}{2+\psi^2} \delta_0(dx) + \frac{\psi^2}{2+\psi^2} \delta_\psi(dx), 
\end{equation}  
with $\psi \in [0,1]$, where $\psi=0$ gives the Kingman coalescent.
\end{itemize}
In the application  examples that follow, we will mainly be concerned with
coalescents \eqref{eq:betameasure} and \eqref{eq:EWoneatom}.  Thus,
the parameter $\alpha$ $(1 < \alpha \leq 2)$ refers exclusively to
coalescent \eqref{eq:betameasure}, and the parameter $\psi$ $(0 < \psi
< 1)$ to coalescent \eqref{eq:EWoneatom}.

For $\alpha$, we focus our attention on the parameter range $1 < \alpha
< 2$, since this corresponds to Beta-coalescents which ``come down
from infinity'', i.e.\ admit a unique most-recent common ancestor even
in the infinite-sample limit \citep{P99,S00}. It can be argued whether this
property is necessary from a biological point of view. However, our
empirical results for real datasets yields so far only estimates of
$\alpha$ between 1 and 2.  In addition, the coalescent process
derived form \cite{schweinsberg03}'s  population model converges to coalescent
\eqref{eq:betameasure} only when $1 \le \alpha < 2$.  For $0 < \alpha
< 1$, the coalescent process is a discrete-time Xi-coalescent
\citep{schweinsberg03}.

\noindent {\bf Related asymptotic results}

In inference, one often relies on asymptotic results.  The asymptotic
behaviour of the site- and allele frequency spectrum of Lambda
coalescents for (very) large sample sizes (as $n \to \infty$), as well
as related asymptotic results concerning the distributions of
tree-lengths or the ``speed of coming down from infinity'', i.e.\ when
the sample size is infinite (a mathematical peculiarity) how long it
takes the ancestral lineages to become finite in number, is an active
field of mathematical research.  The concept of `coming down from
infinity' may yet have real implications for inference.  By way of
example, the point-mass coalescent \eqref{eq:EWoneatom} does not come
down from infinity, which implies that external branches become more dominating as
sample size increases, regardless of the value of $\psi$.  On the
other hand, the beta-coalescent \eqref{eq:betameasure} does come down
from infinity, as does the Kingman coalescent.

Important asymptotic results regarding the SFS are obtained by
\cite[][Theorem 3]{BBL12}, who provide convergence results (in the
almost sure sense), as sample size tends to infinity, of the scaled
SFS for an important class of Lambda-coalescents, including the
beta-coalescent \eqref{eq:betameasure}, as will now be explained.  We
say that the measure $\Lambda$ has (strong) $\alpha$-regular variation
at zero if
$$
\Lambda(dx) = f(x)dx, \quad  
\mbox{ where } f(x) \sim A x^{1- \alpha} \mbox{ as } x \to 0, 
$$
for some $\alpha \in (1,2)$ and $A > 0$. This holds for the
Beta$(2-\alpha, \alpha)$-coalescent from \eqref{eq:betameasure}.
Then, the almost sure asymptotic behaviour of each $\xi^{(n)}_i$ as $n
\to \infty$ is determined by the behaviour of $f(x)$ as $x \downarrow 0$
for a large class of Lambda-coalescents:

\begin{theorem}\citep{BBL12}
\label{thm:bbl}%
Assume that $\Lambda$ has (strong) $\alpha$-regular variation at zero for some $1 < \alpha < 2$ and $A > 0$. Then, for the site frequency spectrum of an $n$-sample, where the underlying genealogy is governed by a Lambda-coalescent and mutation rate $\theta/2$,
$$
\lim_{n \to \infty} \frac{\xi^{(n)}_i}{n^{2-\alpha}} = \frac{\theta}{2} C_{A, \alpha} \frac{(2-\alpha)\Gamma(i+\alpha-2)}{i! \Gamma(\alpha-1)}
= \frac{\theta}{2} C_{A, \alpha} (2-\alpha)\frac{(\alpha-1)\cdots (\alpha+i-3)}{i!}, 
$$
in the almost sure sense, for $i =1, \dots , n$, where
$$
C_{A, \alpha}= \frac{\alpha(\alpha-1)}{A \Gamma (2-\alpha)(2-\alpha)}.
$$
\end{theorem}

We formulate Thm.\ \eqref{thm:bbl} with a mutation rate $\theta/2$
instead of $\theta$ as in \cite{BBL12} in order to be consistent with
\cite{F95} and other literature.  The asymptotic results should be of
direct relevance for us in the case of large sample sizes.  However,
the rate of convergence of the asymptotic results in terms of sample
size $n$ and the parameter $\alpha$ seems to be unresolved.  The issue
of rate of convergence will be investigated below using simulations.


From a statistical point of view an ideal result would be an
(asymptotic) result for the joint distribution of the normalized SFS
for general Lambda-coalescents, and this may well soon be within
reach. For the Kingman case, this has already been achieved: There is
asymptotic normality for the joint distribution of tree-lengths.
\cite{KS13} obtain an asymptotic normality result for the scaled
branch lengths associated with the Kingman coalescent, from which
Kersting and Stanciu deduce that the asymptotic joint distribution of
the frequency spectrum is independent Poisson.





As a partial analogue, \cite{KSW12} prove convergence of the total
length of {\em external} branches of the beta-coalescent
\eqref{eq:betameasure} in distribution to a stable random variable
with index $\alpha$.  Their result indicates that the joint
distribution of the branch lengths associated with a Lambda-coalescent
will not be asymptotically normal, but might instead be a multivariate
stable distribution.  However, the rate of convergence is not known in
the results of \cite{KSW12}.



\begin{center}
{\sc The expected value, variance and covariance of derived mutations in the SFS associated with a Lambda-coalescent }
\end{center}
In this section we compute the expected value $\EE{\xi^{(n)}_i}$, the
variance $\Var{\xi^{(n)}_i}$, and the covariance
$\Cov{\xi^{(n)}_i}{\xi^{(n)}_j}$  of the SFS associated
with Lambda coalescents.  While exact solutions for a finite sample
size $n$ appear quite hard to obtain due to the multiple-merger
property of Lambda-coalescents, we derive recursions for these
quantities.    \cite{F95} (cf.\ Theorem
\ref{thm:Fu} in Appendix for reference) obtained closed-form
expressions for these quantities when associated with the Kingman
coalescent.

As above, let $\xi^{(n)}_i$ denote the random number of derived
mutations in $i$ copies in a $n$ sample with mutation rate $\theta/2$
and genealogy governed by a Lambda-coalescent.  Before stating our
first result, we recall some simple properties of the block-counting
process associated with a Lambda-coalescent from \cite{BB08}.  The
block-counting process $(Y_t)$ simply counts the number of ancestral
lineages present each time.  Recall the corresponding rates
$\lambda_{n,k}$ from \eqref{eq:lambdarates}.  When we refer to number
of `leaves' $(n)$, we will mean the initial sample size. Thus, a
`leaf' will refer to one sampled (DNA) sequence.  Thus, by using
leaves rather than sequences we emphasize that we are thinking of the
sequences as vertices in a graph describing the ancestral relations of
the leaves.  Let $(Y_t)_{t \ge 0}$ be the block counting process of
our Lambda-coalescent, which is a continuous-time Markov chain on $\N$
with jump rates
\[ 
q_{ij} = {i \choose i-j+1} \lambda_{i, i-j+1}, \quad i > j \ge 1.
\]

The total jump rate away from state $i$ is
$-q_{ii} = \sum_{j=1}^{i-1} q_{ij}$. 
By
\begin{equation} 
\label{def:pij}
p_{ij} := \frac{q_{ij}}{-q_{ii}}
\end{equation}
we denote  the transition probabilities of the embedded discrete  skeleton chain. Let 
\begin{equation}%
g(n,m) := \E_n\left[ \int_0^\infty \one_{( Y_s = m )} \, ds \right] 
\quad \mbox{for $\quad n \ge m \ge 2$} 
\end{equation}%
be the expected amount of time that $Y_t$, starting from $n$, spends in $m$. Decomposing according to the first jump of $Y_t$ gives a recursion for $g(n,m)$, 
\begin{equation}
\label{eq:grec} 
g(n,m)  =  \sum_{k=m}^{n-1} p_{nk} 
g(k,m), \quad 
n > m \ge 2, \quad \mbox{ and } \quad 
g(m,m) =  \frac{1}{-q_{mm}}, \quad m \ge 2.  
\end{equation} 


\textbf{The expected value $\EE{\xi^{(n)}_i}$ of derived mutations} 

For $1 < k \le n$,  as in \cite{F95}, an important quantity will be  the probability
$${p^{(n)}[k,b]}, \quad \mbox{ for } b\in \{1, \dots, n-(k-1)\}, $$ 
that {\em in a $\Lambda$-coalescent starting from $n$ leaves,
  conditioned that there are at some point in time exactly $k$
  branches, a given one of these $k$ branches (e.g.\ the first, if we
  think of some ordering) subtends exactly $b$ leaves} (see
Figure~\ref{fig:illustr1}). 


%

\begin{figure}
  \caption{Illustration of the occurrence of an event with probability
    $p^{(n)}[k,b]$, with $k = 2$, and $2 \le b = n - 3$.  Ancestral
    lineage labelled as (2,1), present when there are only two active
    ancestral lineages, is ancestral to $b$ leaves ($b$ sampled
    sequences).  In this example, the first merger in the genealogical
    history of the $n$ leaves is a 3-merger not involving any of the
    encircled $b$ leaves. The second event is a $b$-merger with $b$
    edges subtending the encircled $b$ leaves merging to block $(2,1)$
    (square).  A `level' refers to the values of the lineage-counting
    process $Y_t$.  }
  \label{fig:illustr1}%
  \vspace{-1cm}
\includegraphics[width=6.5in,height=8in]{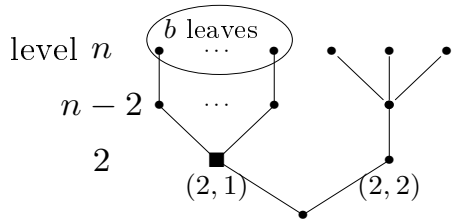}
\vspace{-15cm}
\end{figure}%

%

%



A recursion for $p^{(n)}[k,b]$ (\ref{eq:rnk}), given in Appendix (prop.\ \ref{prop:r}), can be solved numerically, separately for each $k$.   
Proposition~\ref{prop:r} allows us to derive a recursion for the
expected frequency spectrum associated with Lambda-coalescents, given in the following proposition.

\begin{prop}
\label{prop:Exi}
Under the above assumptions, we have, for $i = 1, \dots, n-1$,
$$
\E\left[\xi_i^{(n)}\right] = \frac{\theta}{2}\sum_{k=2}^{n-i+1} {p^{(n)}[k,i]} \cdot k \cdot g(n,k).
$$
\end{prop}
The proof proceeds by decomposing according to the `level' (the
values taken by the lineage-counting process $Y_t$) of the coalescent
tree where a mutation arises - see Appendix for details.
Proposition~\ref{prop:Exi} and its proof are a natural extension of
the arguments of \cite{GT98} to the multiple merger case.


The expression in Prop.\ \eqref{prop:Exi} can be used to define a
quantity $(\varphi_n(i))$, the  `expected normalized frequency spectrum',  which describes the shape of the genealogy and 
does not depend on $\theta$,  as follows.  Let $B_i^{(n)}$ denote the
random total length of branches subtending $i$ leaves, for $1 \le i <
n$.  Then, clearly, $\EE{\xi_i^{(n)}} = (\theta/2)\EE{B_i^{(n)}}$,
with
$$
\EE{B_i^{(n)}} =  \sum_{k=2}^{n-i+1} {p^{(n)}[k,i]} \cdot k \cdot g(n,k).  
$$
If we now let $B^{(n)}$ denote the random total length of the
genealogy started from $n$ leaves, and by $\xi^{(n)} = \xi_1^{(n)} + \cdots + \xi_{n-1}^{(n)}$ we denote the
random total number of segregating sites in a sample of size $n$ gene
copies (DNA sequences), we have
$$
\EE{\xi^{(n)}} = (\theta/2)\EE{B^{(n)}}.
$$
If $T_\ell$ denotes the random time during which  there are $\ell$ active  lineages, for $2 \le \ell \le n$, we have
$$
\EE{B^{(n)}} =  \sum_{\ell=2}^{n} \ell \EE{ T_\ell }.
$$
We define $\varphi$ as 
\begin{equation} 
  \label{eq:expnorm}
\varphi_n(i) := \frac{\E\left[ \xi_i^{(n)} \right] }{\sum_{\ell=2}^{n} \ell \E\big[ T_\ell \big]}
= \frac{\sum_{k=2}^{n-i+1} {p^{(n)}[k,i]}\, \cdot \, k\, \cdot \, g(n,k)}{%
\sum_{\ell=2}^{n} \ell g(n,\ell)}. 
\end{equation}
Loosely speaking, the probability that a mutation, under the
infinitely many sites assumption \citep{K69,W75}, with known ancestral types, appears
$i$ times in a sample of size $n$ is $\varphi_n(i)$.  The quantity
$\varphi_n(i)$ is a quotient of expectations which can be interpreted
as size-biasing - it should not be confused with averaging over the
quotient of these quantities, e.g.\ obtained from an empirical study,
which has a different interpretation and will, at least for finite
sample sizes, take different values.  However, one can apply
$\varphi_n(i)$ as the main ingredient in a pseudo-likelihood approach
to infer coalescence parameters, as we do  below.




\subsection*{The covariance $\Covn{}{\xi_i^{(n)}}{\xi_j^{(n)}}$  of derived mutations}
%
%

In order to compute covariances one needs to compute the probability
that two edges (ancestral lineages) drawn at random without
replacement (ordered sampling) from the past, say from $Y$-states $k$
and $\ell \le n$, produce $i$ resp.\ $j$ subtended leaves.  To this
end, we need to distinguish three situations, namely whether our two
edges are taken from the same `layer' (that is, when $Y$ is in the
same state $k$ for both edges), or, if this is not the case, whether
one of our two edges is a descendant of the other (`nested case'), or
not (`unnested case'). In our calculations the leaves are {\em
  unlabelled}.

\paragraph{Case 1: two edges in the same layer}
For $2 \le k \le n$ let ${p}^{(n)}_{\mathrm{eq}}[k; i,j]$ be the probability 
that in an $n$-$\Lambda$-coalescent, conditioned on there being 
$k \ge 2$ ancestral lineages at some time in the past, if we draw two of the $k$ lines 
at random without replacement from these $k$ (ordered sample), the 
first drawn line has $i$ and the second $j$ subtended leaves (see Figure~\ref{fig:peqill}). 
A recursion for $p_{\mathrm{eq}}^{(n)}$ (prop.\ \ref{prop:1}) and proof are given in Appendix.




\begin{figure}
  \caption{Illustration of the occurrence of the event with probability $p^{(n)}_{\mathrm{eq}}[k;i,j]$ with $i = j = 1$.  The ancestral lines $(k,1)$ and  $(k,2)$ are each ancestral to one leaf.  Level $k$ is reached in one  $n-k+1$-merger not involving the two leaves  subtended by $(k,1)$ or $(k,2)$.        }
  \label{fig:peqill}%
   \vspace{-1cm}
\includegraphics[width=6.5in,height=8in]{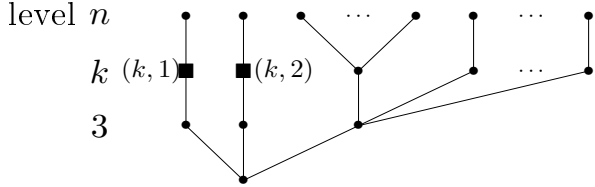}
\vspace{-16.5cm}

\end{figure}%

%


\paragraph{Case 2: Two edges in different layers, unnested case} 

For $2\le k < \ell \le n$ let ${p}^{(n)}_{\mathrm{un}}[k, i; \ell ,j]$ be the
probability that in an $n$-$\Lambda$-coalescent, conditioned on there
being $k$ lines at some time in the past and $\ell$ lines at some
(other) time in the past, if we draw one of the $k$ lines at random
and independently one of the $\ell$ lines at random (ordered sample),
the first drawn line has $i$ and the second $j$ subtended leaves, and
the second is not a descendant of the first (see
Figure~\ref{fig:punillustr}).  A recursion for $p_{\mathrm{un}}^{(n)}$
(prop.\ \ref{prop:2}) is obtained in a similar way as for
$p_{\mathrm{eq}}^{(n)}$, and is given in Appendix along with a proof.

%

\begin{figure}
  \caption{Illustration of the occurrence of the event with
    probability $p^{(n)}_{\mathrm{un}}[k,i;\ell,j]$.  Edge $(\ell,1)$
    in level $\ell$ subtends $j$ leaves, while edge $(k,1)$ in level
    $k$ subtends $i$ leaves; $2 \le k < \ell \le n$, and the two encircled 
    groups of $i$ and $j$ leaves are distinct.  Level $\ell$ is reached in one merger including all of the $j$ leaves.   The following merger  includes all of the  $i$ leaves.     }
  \label{fig:punillustr}%
  
   \vspace{-1cm}
\includegraphics[width=6.5in,height=8in]{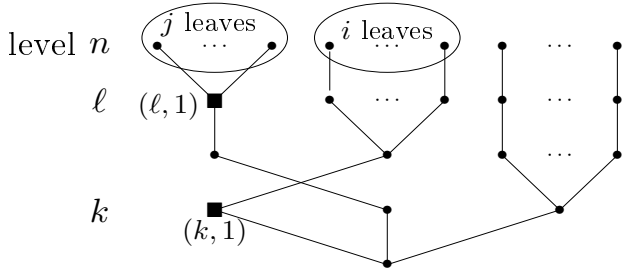}
\vspace{-15cm}

\end{figure}%

%


\paragraph{Case 3: Two edges in different layers, nested case.} 

For $2\le k<\ell\le n$ let ${p}^{(n)}_{\mathrm{ne}}[k, i; \ell,j]$ be
the probability that in an $n$-$\Lambda$-coalescent, conditioned on
there being $k$ lines at some time in the past and $\ell > k$ lines at
some more recent time in the past, if we draw (ordered sample) one of
the $k$ lines at random ($(k,1)$, say, see
Figure~\ref{fig:punillustr2}) and independently one of the $\ell$
lines ($(\ell, 1)$, say, see Figure~\ref{fig:punillustr2}) at random,
the first drawn line has $i$ and the second $j$ subtended leaves, and
the second $(\ell, 1)$ is a descendant of the first ($(k,1)$; so
necessarily $i \ge j$, otherwise the probability is zero; see
Figure~\ref{fig:punillustr2}). The recursion for
$p_{\mathrm{ne}}^{(n)}$ is given in proposition \ref{prop:3} in
Appendix, with the proof following similar arguments to  proofs of propositions \eqref{prop:1} and \eqref{prop:2}.



\begin{figure}
  \caption{Illustration of the occurrence of the event with
    probability $p^{(n)}_{\mathrm{ne}}[k,i;\ell,j]$.  Edge $(\ell,1)$
    in level $\ell$ subtends the encircled subset of $j$ leaves, while edge $(k,1)$ in level
    $k$ subtends $(\ell,1)$ and the larger encircled set  of $i$ leaves necessarily containing the set of $j$ leaves.  }
  \label{fig:punillustr2}%
  
   \vspace{-1cm}
\includegraphics[width=6.5in,height=8in]{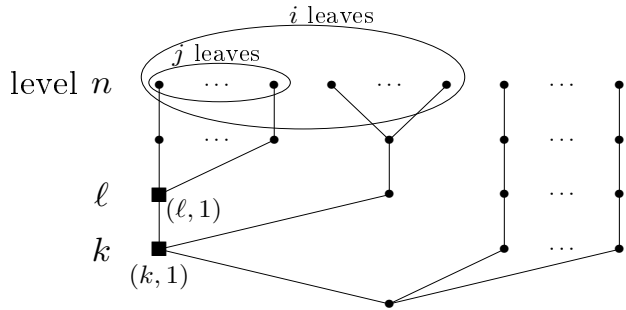}
\vspace{-15cm}

\end{figure}%

%
Unfortunately,   the computational complexity of the recursions for
$p^{(n)}_{\mathrm{eq}}$ \eqref{eq:p.eq.rec}, $p^{(n)}_{\mathrm{un}}$ \eqref{eq:p.un.rec}, and
$p^{(n)}_{\mathrm{ne}}$ \eqref{eq:p.ne.rec}, as they stand, is $O(n^5)$, thus restricting
numerical computations to relatively small values of $n$.

\paragraph{Second moments and covariance of the frequency spectrum}
Now we can put our recursions for $p^{(n)}_{\mathrm{eq}}$,
$p^{(n)}_{\mathrm{un}} $, and $p^{(n)}_{\mathrm{ne}}$, together to
obtain the covariances $\Cov{\xi_i^{(n)}} {\xi_{j}^{(n)}}$.  Our main
result on the site frequency spectrum associated with
Lambda-coalescents, an expression for
$\E\left[\xi_i^{(n)}\xi_j^{(n)}\right]$ in terms of
$p^{(n)}_{\mathrm{eq}}$, $p^{(n)}_{\mathrm{un}} $,
$p^{(n)}_{\mathrm{ne}}$, $p^{(n)}[k,i]$, and $g(n,m)$, will now be
stated - a proof is in Appendix (section~\ref{sec:proofcov}).
\begin{theorem}
\label{thm:cov}
For $1 \le i, j < n$, $2 \le i + j \le n$, we have
\begin{align*}
\E&\big[\xi^{(n)}_i \xi^{(n)}_j \big] \\
&= \, \frac{\theta^2}{4} \sum_{k=2}^n k(k-1) 
{p}^{(n)}_{\mathrm{eq}}[k; i,j]
\frac{g(n,k)}{g(k,k)} 
\frac{2}{(-q_{kk})^2} \\
& \,\;\; + \one_{(i=j)} \sum_{k=2}^n k {p^{(n)}[k,i]} \frac{g(n,k)}{g(k,k)} 
\Big( \frac{\theta}{2} \frac{1}{-q_{kk}} + 
\frac{\theta^2}{4}\frac{2}{(-q_{kk})^2} \Big)\\
& \;\; + 
\frac{\theta^2}{4} \sum_{k=3}^n \sum_{\ell=2}^{k-1} k \, k' 
\frac{{p}^{(n)}_{\mathrm{un}}[k,i ; \ell,j] 
+ {p}^{(n)}_{\mathrm{ne}}[k,i ; \ell,j] + 
{p}^{(n)}_{\mathrm{un}}[k,j ; \ell, i] 
+ {p}^{(n)}_{\mathrm{ne}}[k,j ; \ell, i]}{(-q_{kk})(-q_{\ell\ell})} \\
& \phantom{MMMMMMMMMMMMMMMMMMMMMMMMMMM} \times
\frac{g(n,\ell)}{g(\ell,\ell)} \frac{g(\ell,k)}{g(k,k)}.
\end{align*}
\end{theorem}

In Figures~\ref{cov1} and \ref{cov2} we graph the covariances
\eqref{eq:EBiBj} of the branch lengths $B_i^{(n)}$ and $B_j^{(n)}$,
where $B_i^{(n)}$ denotes the total length of branches subtending $i$
of $n$ leaves, and the coalescent process is the Beta$(2-\alpha,
\alpha)$-coalescent.  
The variances and covariances become completely dominated by the
variance of the length of external branches when multiple mergers are more prominent
in the genealogy, i.e.\ when $\alpha$ is close to 1
(Figure~\ref{cov1}).  This means that the variances and covariances of
the SFS will be dominated by the variance in the number of singletons,
when associated with Lambda-coalescents, and multiple mergers dominate
the genealogy, as would occur in a high fecundity population with
frequent occurrence of large offspring number events.  In the beta
coalescent \eqref{eq:betameasure}, large offspring number events
become more prominent as $\alpha$ approaches 1, and less prominent as
$\alpha$ approaches 2.  And even when $\alpha$ is more `modest' at
$1.5$, the covariances are still quite small compared to the
variances.  Figure~\ref{cov2} shows the covariances on their own
scale, i.e.\ with the variances and the diagonal covariances
$\Cov{B_i^{(n)}}{B_{n-i}^{(n)}}$ set at zero.  These results suggest
that one may model the site-frequency spectrum of an
$n$-$\Lambda$-coalescent as composed of uncorrelated (but not
independent) components $\xi_i^{(n)}$ for $n$ large enough.

For fixed $n$, the covariances between $B_i^{(n)}$ and $B_{n-i}^{(n)}$
seem to be larger than the covariances off the diagonals. As for the
Kingman case, this can be explained by considering the first branching
event in the coalescent history forward in time after the root. This
has substantial probability to be a binary split, say in subfamilies
of size $i$ and $n-i$, which yields a positive correlation between
$B_i^{(n)}$ and $B_{n-i}^{(n)}$.


\clearpage
\pagebreak
\newpage

\begin{figure}[h!]
  \caption{The covariance $\Cov{B_i^{(n)}}{B_j^{(n)}}$ $(C)$ for sample size $n = 40$ for the Beta$(2-\alpha, \alpha)$-coalescent with $\alpha$ as shown.  The lower  panels show only the  covariance terms,  with the variance terms  $\Var{B_i^{(n)}}$ set equal to zero.  The covariances range between  $-0.094$ and $14.85$ when $\alpha = 1.05$, and between  $-0.045$  and   $2.06$ when $\alpha = 1.5$.  One has $\Cov{\xi_i^{(n)}}{\xi_j^{(n)}} = (\theta^2/4)\Cov{B_i^{(n)}}{B_j^{(n)}}$   for $i \neq j$, and  $\Var{\xi_i^{(n)}} = (\theta/2)\EE{B_i^{(n)}} + (\theta^2 /4)\Var{B_i^{(n)}}$.    }
  \label{cov1}
 \vspace{-1cm}
  \includegraphics[width=7in,height=7in]{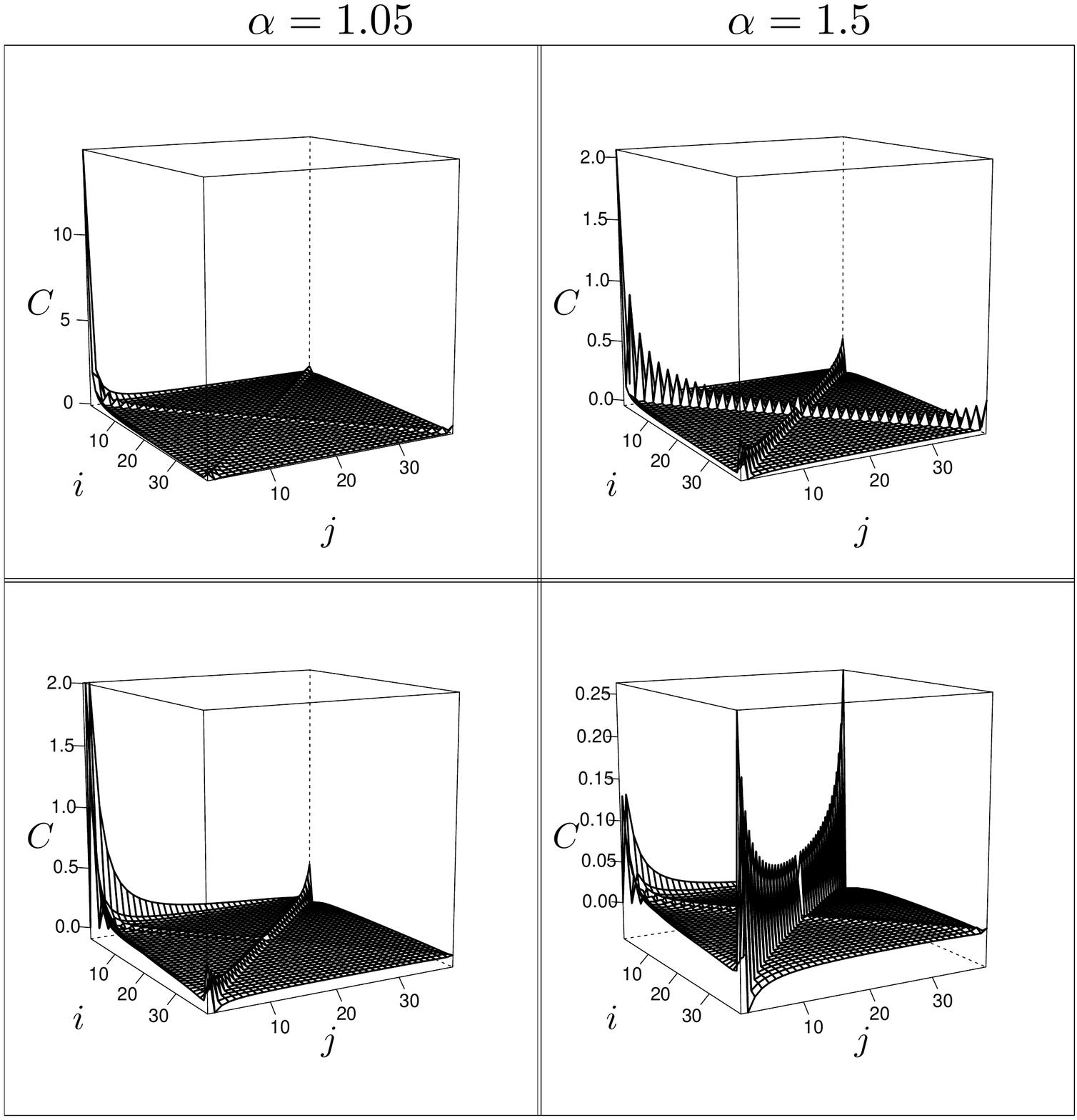}
  \vspace{-3cm}
  \end{figure}%

\clearpage
\pagebreak
\newpage


\begin{figure}[h!]
  \caption{ The covariance $\Cov{B_i^{(n)}}{B_j^{(n)}}$ $(C)$ for sample size $n = 40$ for the Beta$(2-\alpha, \alpha)$-coalescent with $\alpha$ as shown, and with covariances  $\Cov{B_i^{(n)}}{B_{n-i}^{(n)}}$ and the variance terms set equal to zero.    }
\vspace{-1cm}
\vspace{-1cm}
\label{cov2}
\includegraphics[width=7in,height=7in]{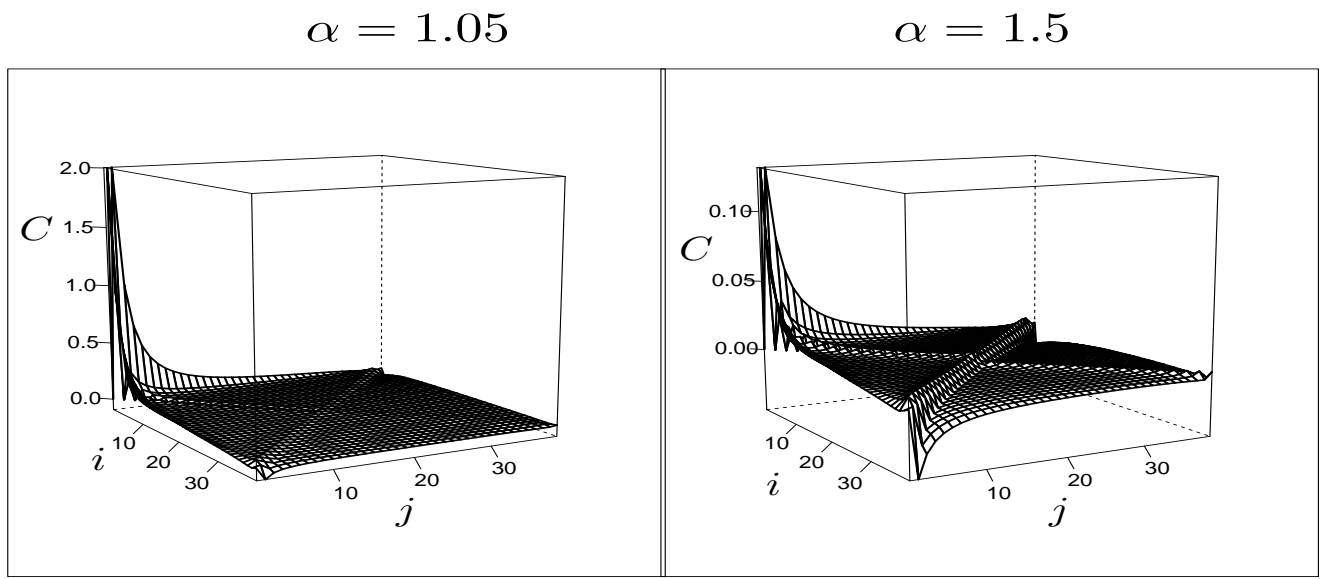}
  \vspace{-3cm}
  \end{figure}%

\clearpage
\pagebreak
\newpage

Extending our results to the case of Xi-coalescents and multiple loci
is a natural follow-up task, but lies beyond the scope of the current
discussion and will be part of future research.

\medskip

\begin{center}%
{\sc  Simulation studies and data analysis}
\end{center}%

{\bf Reliability of the  asymptotic results of  \cite{BBL12}}

The asymptotic results obtained by \cite{BBL12} (see
Theorem~\eqref{thm:bbl}) can be used to obtain estimates of $\alpha$
provided the asymptotic results are `close' to the true value,
\emph{and} one has an estimate of $\theta$.  Ideally, such inference
should take into account correlation between loci.  Now we address the
question how large the sample must be for the SFS to be well
approximated by the asymptotic results of \cite{BBL12}, and, indeed,
if one can trust the asymptotic results for all values of $\alpha$.

Figure~\ref{fig:bblfitnmed} compares the asymptotic results in
Theorem~\eqref{thm:bbl} to simulated values for sample size $n=500$
(left column) and $n=1000$ (right column), and with $\alpha$ varying
as shown.  Even though we have recursions to obtain the exact expected
values for finite $n$, we were not able, with present computer power,
to compute covariances for $n \ge 100$.  The estimated expected values
and the ones from the recursions are in agreement (results not shown).
Even for $n=10^4$ (Figures~\ref{fig:bblfitnbig1} and
\ref{fig:bblfitnbig2}), which would be a rather large sample, the
asymptotic results (grey bars) match the simulated ones (white bars)
rather poorly when $\alpha < 1.5$, but the correspondence improves as
$\alpha$ increases, and at $\alpha = 1.5$ the match is already quite
good.  For more common sample sizes $n \le 10^3$
(Figures~\ref{fig:bblfitnsmall} and \ref{fig:bblfitnmed}) the match is
not very good for $\alpha \le 1.5$.  Relying on asymptotic results to
estimate $\alpha$ might thus give misleading results, even when sample
size is quite large, as demonstrated in Figure~\ref{fig:bblsingles}
which compares the asymptotic results for the singletons to the scaled
(with $n^{\alpha - 2}$) exact expected values as sample size $n$ and
$\alpha$ vary as shown.  The poor fit of the asymptotic results to
exact values of the SFS, at least for lower values of $\alpha$,
 suggest that exact results for the SFS, if only in the form
of recursions, may be a better choice for inference methods.

\clearpage
\pagebreak
\newpage

\clearpage
\pagebreak
\newpage


\begin{figure}[h!]
  \caption{The results of \cite{BBL12}  (Thm.\ \ref{thm:bbl}) regarding convergence of the scaled unfolded frequency spectrum  (grey bars) compared to  simulations (white bars) varying over $\alpha$ as shown,  with $\theta = 1$.     Panels in left column are for sample size  $n=500$, panels in right column are for sample size $n=1000$. The whiskers represent standard deviation.     }
  \label{fig:bblfitnmed}%
  \vspace{-1cm}
\hspace{-1in}\includegraphics[width=8.5in,height=9in]{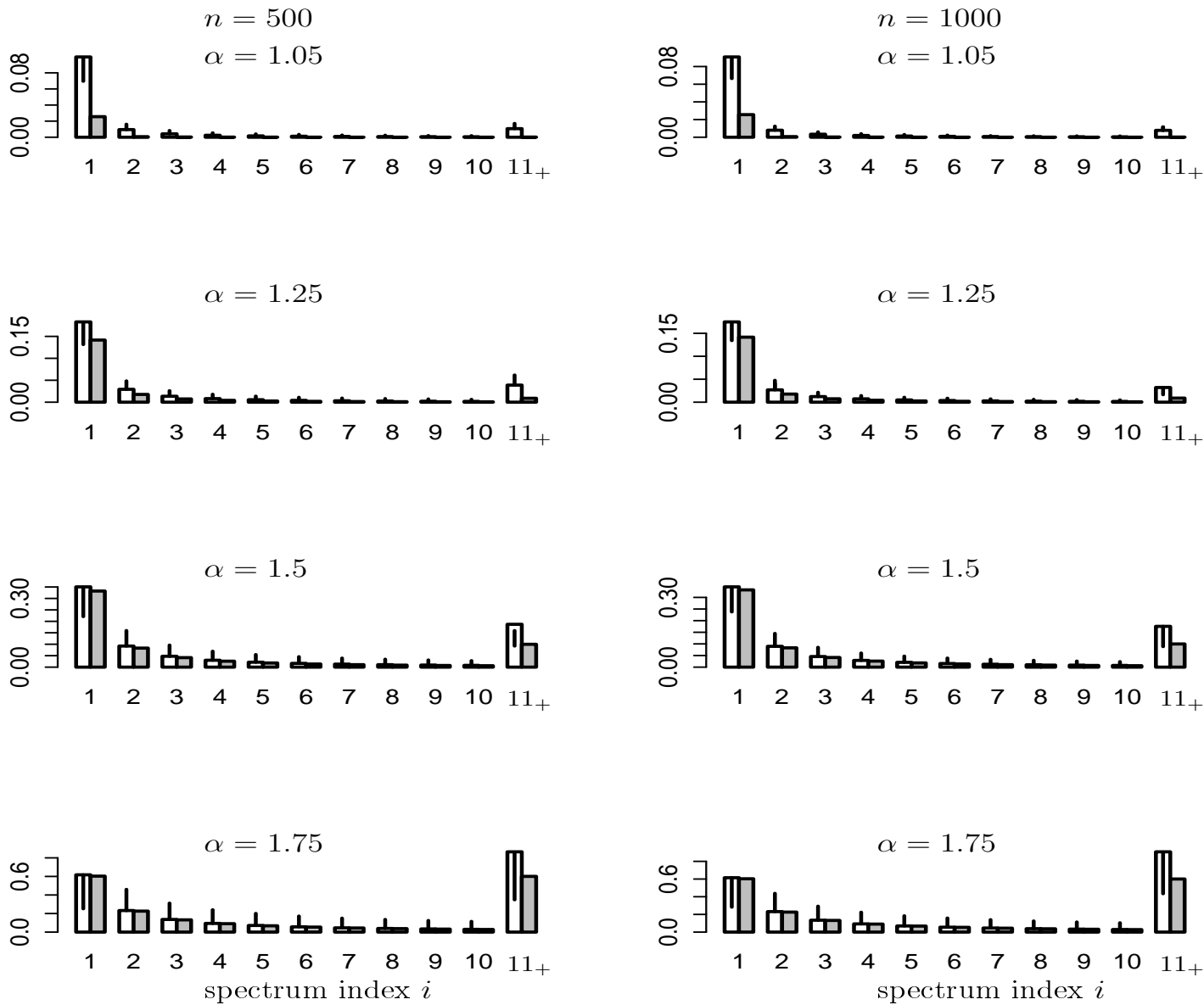}
\vspace{-1cm}
\end{figure}%

\clearpage
\pagebreak
\newpage

\begin{figure}[h!]
  \caption{The asymptotic results of \cite{BBL12} (Thm.\
    \ref{thm:bbl}) for the scaled singletons (lines) compared to the
    corresponding scaled exact expected values (symbols) as a function of sample
    size $n \in \{10,20,50,100,200,300,500,1000,2000,10000\}$ on log-scale,  and varying over $\alpha$ as shown.  Thus, the solid line and circles
    refer to $\alpha = 1.05$, the short-dashed line and squares to $\alpha = 1.25$, and so  on.  }
  \label{fig:bblsingles}%
  \vspace{-1cm}
\includegraphics[width=5in,height=5in]{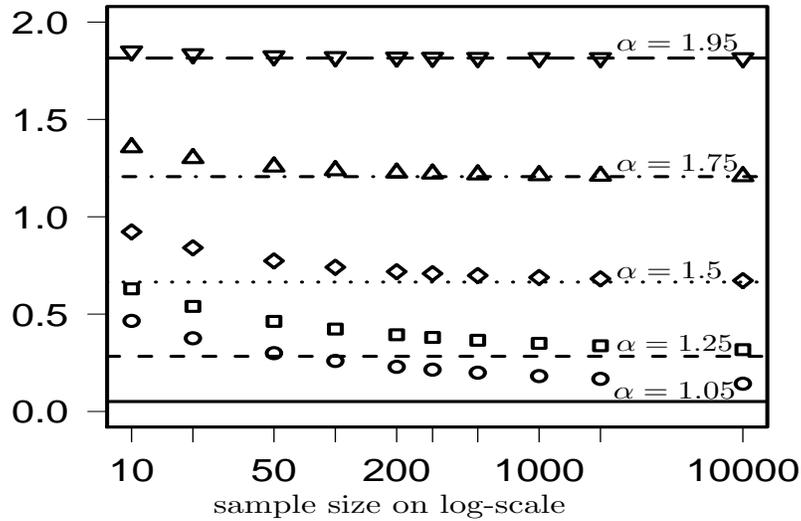}
\vspace{-1cm}
\end{figure}%

\clearpage
\pagebreak
\newpage

\begin{center}
{\sc Inference methods and data analysis}
\end{center}

  Before we discuss our results on the analysis of the
Atlantic cod data, we introduce our two inference methods, comparing
observed and expected values using an $\ell^2$-distance, and a
pseudo-likelihood approach.

{\bf The $\ell^2$-distance}


The simplest way to find a value of $\alpha$ or $\psi$ that gives the
best fit to the data would be to compare the expected values of the
frequency spectrum to the observed one, and report the parameter value
that minimizes the sum of squares \eqref{eq:ltwodist}.  In order to avoid having to
jointly estimate the mutation rate, we can define the scaled frequency
spectrum $\zeta_i^{(n)}$ as  
$$
\zeta_i^{(n)} :=  \frac{\eta_{i}^{(n)} }{ \sum_{i=1}^{\lfloor n/2 \rfloor}\eta_i^{(n)} } 
$$
The observed $\zeta_i^{(n)}$ are compared to the ratio of
expected values   
$$
r_i^{(n)} :=  \frac{\EE{\eta_{i}^{(n)} } }{\sum\limits_{i=1}^{\lfloor n/2 \rfloor}\EE{\eta_{i}^{(n)}}  }.
$$
Although clearly one should compare $\zeta_i^{(n)}$ to its actual
expected value, simulation results indicate that our approximation is
quite good (results not shown).   A natural distance measure to compare $\zeta_i^{(n)}$ and $r_i^{(n)}$ is the
$\ell^2$-distance
\begin{equation}%
\label{eq:ltwodist}%
\ell^2 = \sqrt{ \sum_{i=1}^{\lfloor n/2 \rfloor}\left( \zeta_i - r_i\right)^2 }
\end{equation}%
A similar distance measure is the $G_\xi$ statistic proposed by
\cite{F96}, in which each term in the sum is weighted by the variance
of $\xi_i^{(n)}$ $\left( \eta_i^{(n)} \right)$.  The computational
cost of computing the variances for a Lambda-coalescent prevents us
from considering an identical statistic for Lambda-coalescents.  A
drawback of \cite{F96}'s $G_\xi$ statistic is that it depends on the
mutation rate $\theta$.  Ideally, one would want to apply a statistic
that did not depend on $\theta$.  Instead, a statistic that reflects
the topology of the underlying genealogy is desirable, since the
processes that one usually tries to learn about in population genetics
affect the topology of the genealogy of a sample in one way or
another.  Determining statistical significance levels of the
$\ell^2$-distance will be postponed for now; however one may
apply the approach of \cite{F96} to determine significance levels and
statistical power.

{\bf  The pseudo-likelihood approach}




A more fundamental statistical inference procedure than simply
applying the $\ell^2$-statistic is a likelihood inference.  Ideally,
in our case, to construct a likelihood function, one would have a way
of writing down the probability of observing a mutation $i$ times in a
sample in terms of the coalescence parameters.  However, these
probabilities depend on the coalescence parameters in a way that is
hard to make explicit, and therefore we approximate them with the
quantity $\varphi_n(i)$ \eqref{eq:expnorm}.  We interpret the observed
site-frequency spectrum as an observation from a multinomial
distribution, in which the probability of each class is estimated
using $\varphi_n(i)$.  For $1 \le i < n$, $\varphi_n(i)$ is the
expected number of mutations in $i$ copies scaled by the expected
total number of segregating sites, which corresponds to the expected
total number of branches subtending $i$ leaves scaled by the expected
total size of the genealogy.  Since we can only compute $\varphi_n(i)$
numerically, Fisher Information cannot be computed easily, and it
would not be a priori clear if the variance and confidence bounds
derived therefrom via classical statistical theory would be
valid. Finally, we refer to our likelihood approach as a
pseudo-likelihood inference.

Now we state our pseudo-likelihood function.  Let $s$ denote the total
number of segregating sites, and let $\boldsymbol{s} = (s_1, \ldots,
s_k, s_{k+})$ the observed folded frequency spectrum, where
$s_{k_{+}}$ denotes the collapsed class.  Write $\varphi_i :=
\varphi_n(i)$, and define $\varphi_{k_+} = 1 - \varphi_1 - \cdots -
\varphi_k$.  The pseudo-likelihood function $L\left(\pi |
  \boldsymbol{s} \right)$ takes a multinomial form, in which $\pi$ denotes the
appropriate coalescence parameter ($\psi$ or $\alpha$):
\begin{equation}%
\label{eq:likelihoodfunc}
L\left(\pi | \boldsymbol{s} \right) =  \binom{s}{s_1 \cdots s_k s_{k_+}}\varphi_1^{s_1}\cdots\varphi_k^{s_k}\varphi_{k_+}^{s_{k_+}}.
\end{equation}%
The form of $L$ can be seen from the way we view the data
$\boldsymbol{s}$: as an observation from a multinomial distribution in
which the probability of each class is estimated with $\varphi_i$. It suffices to consider the pseudo-log-likelihood function   
$$
\ell \left(\pi | \boldsymbol{s} \right) =  s_1\log\left( \varphi_1\right) + \cdots + {s_{k_+}}\log\left(\varphi_{k_+}\right).
$$
As in the case of the statistics $\zeta_i^{(n)}$ and $r_i^{(n)}$, this
pseudo-likelihood approach does not require an estimate of $\theta$.
In Table~\ref{tab:Echi3} we do show that the pseudo-likelihood
performs as it should, when applied to the expected site-frequency
spectrum of different coalescent processes.

Table~\ref{tab:Echi3} shows maximum likelihood estimates of the
coalescence parameters $\alpha$ and $\psi$, when the data are the
expected values of the site frequency spectrum,   obtained for given values of the appropriate
coalescent parameter, denoted by $\pi$ in Table~\ref{tab:Echi3}, and computed using our recursion (prop.\ \ref{prop:Exi}).   
Thus, $\pi = \alpha$ when $\pi \in (1,2)$, and $\pi = \psi$ when $\pi
\in (0,1)$.  For each combination of sample size $n$ and $\pi$, the
log-likelihood values are compared with the ones obtained for the
Kingman coalescent $\log L(0)$.  Highest log-likelihood values are
always obtained for the correct parameter value even for sample size
$n$ as low as ten.  And for modest sample sizes, one should be able to
correctly identify the Kingman coalescent ($\pi = 0$ in Table~\ref{tab:Echi3}).

In Figures~\ref{fig:graphlBBab}--\ref{fig:mlekingman} in Supporting Information we estimate the
distribution of the maximum-likelihood estimator for different values
of $n$, $\theta$, and the coalescence parameters $\alpha$ and $\psi$.
As one expects, the performance of the pseudo-likelihood is better for
larger values of $n$ and $\theta$.

A natural question to address is how well do the probabilities
$\varphi_n(i)$ \eqref{eq:expnorm} actually fit with the correct values $\EE{\zeta_i^{(n)}}$
In Figures~\ref{fig:errfit0} and \ref{fig:errfit1} in Supporting Information we
address this issue via simulations. The agreement between $p_i$ and
$\overline{R_i}$ is best when the  number $n$ of leaves (sample size) is
large, and $\alpha$ closer to 2 than to 1.

The C code written to carry out the computations is available upon
request.

\clearpage
\pagebreak
\newpage

\begin{table}[H!]
  \caption{Pseudo-log-likelihood values and maximum likelihood estimates ($\hat{\alpha}$ and $\hat{\psi}$) of coalescent parameters $\alpha$ and $\psi$  for the  pseudo-likelihood \eqref{eq:likelihoodfunc}  with  $p_i := \EE{L_i}/\EE{L}$. By $L(0)$ we denote the log-likelihood associated with the Kingman coalescent.  Data is the expected unfolded frequency spectrum associated with a particular coalescent with parameter $\pi$.  Values associated with the $\Lambda_\psi$ coalescent \eqref{eq:EWoneatom} were computed for the range $[0.02, 0.96]$, and for $\Lambda_\alpha$ \eqref{eq:betameasure} in the range $[1.02, 1.96]$.  One  computes $\log L(0) = \sum_{i=1}^{n-1}(1/i)\log((2/i)/\EE{L_K} )$, $\EE{L_K} = 2\sum_{i=1}^{n-1}1/i$.            }
  \label{tab:Echi3}%
  \begin{tabular}{lllllllll}
    \hline \\
    $n$ &  $\pi$ &  $\hat{\alpha}$ & $\hat{\psi}$  &  $\log L(\hat{\alpha})$ & $\log L(\hat{\psi})$ & $\log L(0)$  \\ 
    \hline
    10 & 0 & $1.96$ & $0.02$ &  $-5.404$  & $-5.405$  &  $-5.404$ \\
   100 &  &  $1.96$ & $0.02$ & $ -19.027$ & $ -19.176$ & $ -19.021$ \\
   1000&  &  $1.96$ & $0.02$ & $ -38.868$ &  $-43.281$  & $ -38.847$ \\
   10 & $1.05$ & $1.06$ & $0.30$ & $-5.122$ & $-5.159$ &  $-5.583$  \\
   100& & $1.06$ & $0.12$  &   $-26.179$ & $-27.483$ & $ -35.3964 $ \\
   1000&& $1.04$ &$0.04$ &  $-106.80$ & $-120.525$ & $ -204.1335$ \\
   10 & $1.5$ & $1.50$ & $0.14$ & $-5.410$ & $-5.419$ & $ -5.508$ \\
   100& & $1.50$ & $0.04$ &  $-23.731$ & $-24.163$ & $ -25.272$ \\
   1000&& $1.50$ & $0.02$ & $-71.143$ & $-78.175$ &  $ -81.3238$ \\
   10 & $0.01$ & $1.96$ & $0.02$  & $-54110.92$ & $-54113.630$  & $-54113.59$ \\
   100&        & $1.90$ & $0.02$  &  $-194644.8$ & $-194884.834$  & $-194951.4$ \\
   1000&       & $1.62$ & $0.02$ &  $-440572.40$&  $-432453.967$ & $ -466074.3$ \\
   10  &  $0.5$& $1.02$ & $0.5$  &  $-19.219$   &  $-18.378$  & $ -23.28243$ \\
   100 &       & $1.02$ & $0.5$ &  $-85.749$   &  $-63.553$  & $ -207.8592$ \\
   1000&       & $1.02$ & $0.5$ &  $-329.189$  & $-130.413$  & $ -2099.191$ \\
    \hline
    \hline
\end{tabular}\end{table}

\clearpage
\pagebreak
\newpage

{\bf  Atlantic cod mtDNA  data analysis }

An application of the recursion (Prop.\ \eqref{prop:Exi}) for the
exact expected values for the site-frequency spectrum will now be
illustrated using data on Atlantic cod obtained from various
localities in the North-Atlantic \citep{A00,SA03,A04}.  The shallow
gene genealogies observed among the haplotypes sampled, coupled with
very high fecundity of Atlantic cod, lead \cite{A04} to argue that the
Kingman coalescent may not be an appropriate null-model for Atlantic
cod.  Later studies \citep{E11,SBB13} support that conclusion by
finding better match between the data and the beta-coalescent, than
obtained for the Kingman coalescent.  \cite{E11} employ the asymptotic
results of \cite{BBL12} on the site-frequency spectrum obtained from
\cite{A04}, while \cite{SBB13} employ full-likelihood methods on the
sequence data of Atlantic cod datasets of smaller size than in
\cite{A04}.

The Atlantic cod mtDNA  datasets we analyse are summarized in Table~\ref{tab:coddata}.
The sequence data from the Faroe Islands \citep{SA03} consist of a 250
bp segment from the mitochondrial \emph{cytochrome} b gene in addition
to 242bp and 74bp segments for a total of 566bp.  The data from
Greenland and Iceland described in \cite{A00} consist of the 250bp
segment, and so does the combined data  described in  \cite{A04}.

\begin{table}[h!]
  \caption{The folded frequency spectrum for the Atlantic cod data in \cite{A00}, \cite{SA03}, and \cite{A04}. The sample size of each dataset is given in parentheses.   Mutations of  frequency  classes from  11 onwards are lumped together in the $11_+$ class except for the data from Faroe Islands,   in which mutations from seven onwards are lumped into the $7_+$ class. The symbol $s$ denotes the total number of segregating sites.     }%
  \label{tab:coddata}%
  \centering
  \begin{tabular}{lrrrrrrrrrrrr}%
    \hline
    source & $s$& $\eta_1$& $\eta_2$&  $\eta_3$&  $\eta_4$&  $\eta_5$&  $\eta_6$&  $\eta_7$&  $\eta_8$&  $\eta_9$&  $\eta_{10}$&  $\eta_{11+}$ \\
    \hline
    \cite{A00} &   \\
    Greenland (78) & 9&  3& 1& 0& 1& 0& 0& 1& 1& 1& 0& 1 \\
    Iceland (519) &23 &  12&  1&  4&  1&  0&  0&  0&  0&  0&  0&  5\\
    joint (597) & 24 & 11 & 3 & 3 & 1 & 1 & 0 & 0 & 0 & 0 & 0 & 5 \\
    \hline
    \cite{SA03}\\
    Faroe Islands (74) &44  & 23 & 9 & 1 & 1 & 1 & 0 & 9 & 0 & 0 & 0 & 0 \\
    \hline
    \cite{A04} ($1278$) &39 & 17 & 4 & 4 & 2 & 4& 1&0&1&0&0& 6 \\
    \hline
    \end{tabular}%
\end{table}%

In Figures~\ref{fig:rltd_a} and \ref{fig:rltd_p} the
$\ell^2$-distances comparing $\zeta_i^{(n)}$ and $r_i^{(n)}$ obtained
for the different Atlantic cod datasets are compared.  The
corresponding parameter estimates are reported in
Table~\ref{tab:ltwomin}.  For the larger datasets one observes much
`sharper' curves than for the smaller ones as one would expect, since
the larger datasets carry more information.  The difference in
$\ell^2$-distance between the coalescent processes is also biggest for
the larger datasets, with the beta-coalescent showing the best fit for
all the datasets, although the difference between the two
Lambda-coalescents is not always large.  However, based on this simple
parameter estimation method, one would clearly prefer one of the two
Lambda-coalescents over the Kingman coalescent.

\clearpage
\pagebreak
\newpage

\begin{figure}[h!]
  \caption{The $\ell^2$-distance between $\zeta_i$ and $r_i$ for the
    five different Atlantic cod datasets, with expected values
    obtained for the beta-coalescent.  The values associated with the
    Kingman coalescent are obtained for $\alpha = 2$.  The A04
    dataset refers to the one in \cite{A04}.  }
  \label{fig:rltd_a}%
\includegraphics[width=5in,height=5in]{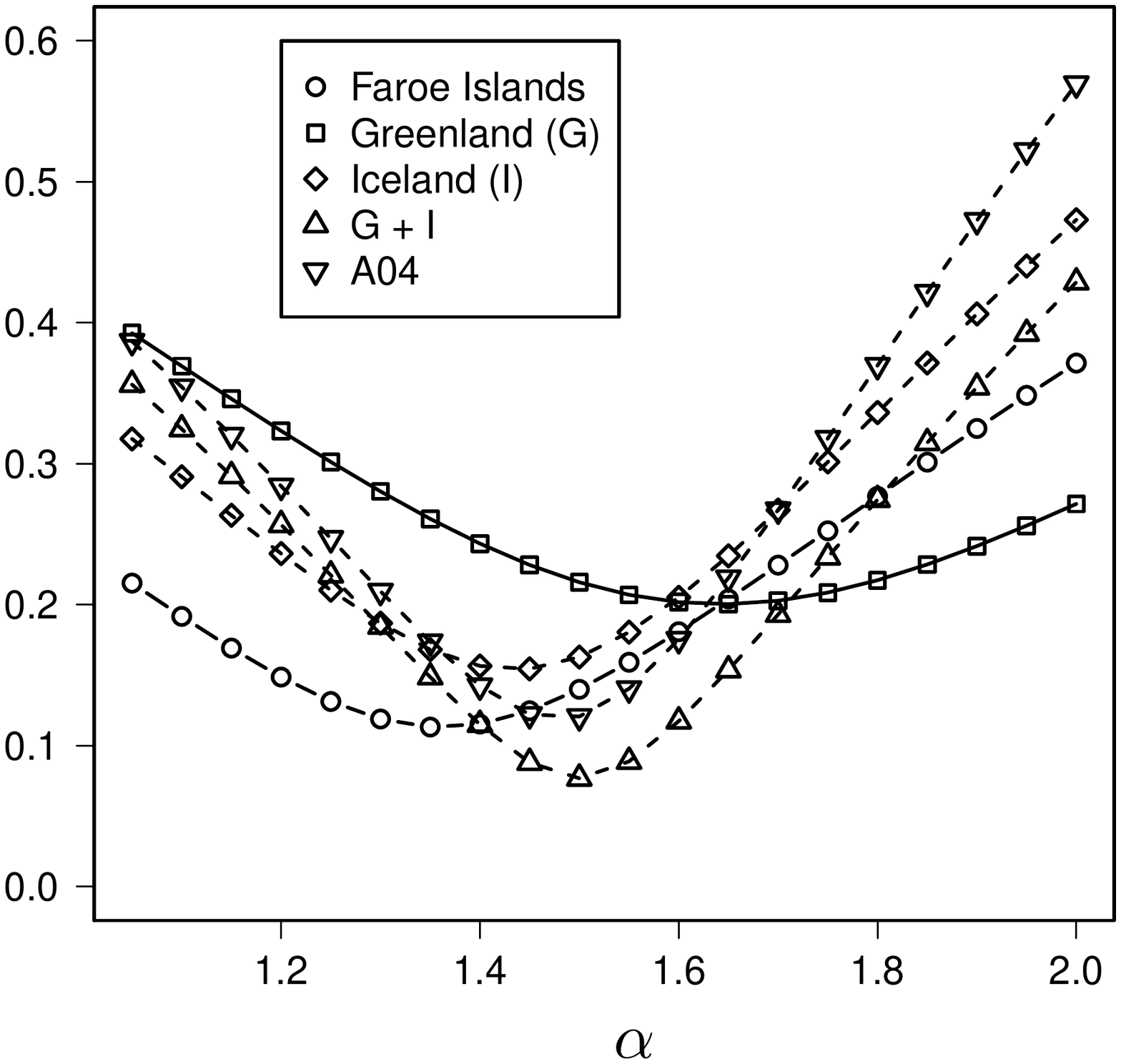}
\end{figure}%

\clearpage
\pagebreak
\newpage

\begin{figure}[h!]
  \caption{The $\ell^2$-distance  between  $\zeta_i$ and  $r_i$ for the five different Atlantic cod datasets, with expected values obtained for the point-mass coalescent.     See Figure~\ref{fig:rltd_a} for explanation of symbols.    }
  \label{fig:rltd_p}%
\includegraphics[width=5in,height=5in]{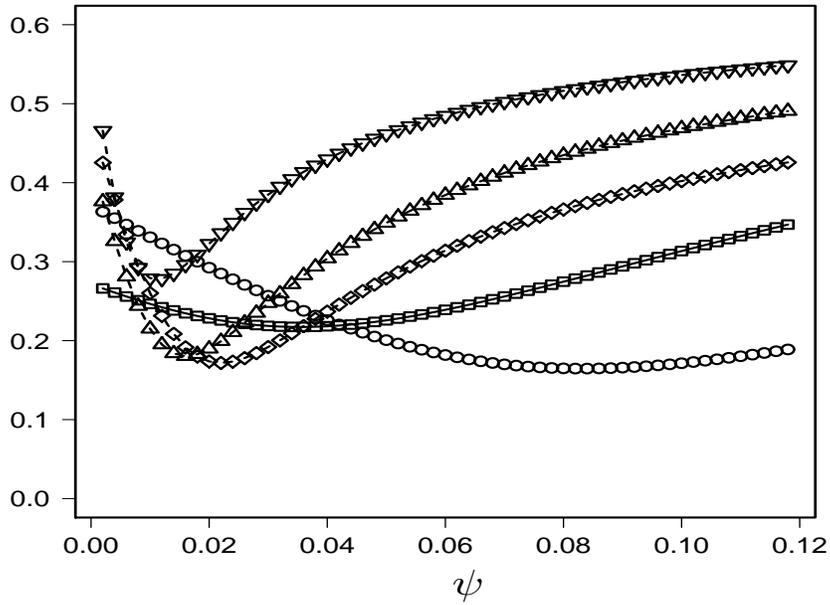}
\end{figure}%

\clearpage
\pagebreak
\newpage

\begin{table}[H!]
\caption{The parameter values minimizing the $\ell^2$ distance.  The $\ell^2$-distance  $\ell^2(0)$   associated with the Kingman coalescent is given for reference. The reference `A04' refers to the large dataset in \cite{A04}.       }
\label{tab:ltwomin}%
\centering
\begin{tabular}{llllll}
\hline
data &  $\hat{\alpha}$ & $\hat{\psi}$ &  $\ell^2(\hat{\alpha})$ &    $\ell^2(\hat{\psi})$ &   $\ell^2(0)$   \\
\hline
Faroe Islands &  $1.35$ & $0.084$ & $0.113$ &  $0.165$  &  $0.371$   \\
Greenland (G) &  $1.65$ & $0.036$ & $0.200$ &  $0.217$  & $0.272$ \\
Iceland (I)  &  $1.45$ & $0.022$ & $0.154$ & $0.172$ & $0.473$ \\
G + I & $1.50$ & $0.016$ & $0.077$ & $0.180$   & $0.429$  \\
A04  & $1.50$ & $0.012$ & $0.121$ &  $0.278$  &  $0.569$  \\
\hline
\end{tabular}
\end{table}

\clearpage
\pagebreak
\newpage

The maximum-likelihood estimates (Table~\ref{tab:mlecod}) agree quite
well with the $\ell^2$-distance approach in Table~\ref{tab:ltwomin}.
\cite{SBB13} apply full-likelihood methods on small Atlantic cod
datasets, including the Faroe Islands dataset \citep{SA03}, and the
Greenland subsample of \cite{A00} to obtain estimates of $\alpha$.
Their estimates $(\hat{\alpha} = 1.5)$ for the Greenland subsample,
and $\hat{\alpha} = 1.3$ for the Faroe Islands sample, agree with our
estimates, in particular the Faroe Islands sample.

\begin{table}[h!]
  \caption{Estimates $(\hat{\alpha}, \hat{\psi})$ based on the pseudo-likelihood \eqref{eq:likelihoodfunc} for the Atlantic cod data in \cite{A00}, \cite{SA03}, and \cite{A04}. The estimate $\hat{\theta}$ of $\theta$ is the one associated with the Beta$(2-\alpha,\alpha)$-coalescent. The maximum pseudo-log-likelihood values $\log L(\cdot)$, along with   the corresponding value $\log L(0)$ associated with the Kingman coalescent, are given.   The $\ell^2$ distances between the observed and expected site-frequency spectrum  computed for the maximum pseudo-likelihood estimates  are given for the  beta-coalescent $\left(\ell_B^2\right)$, and the Kingman coalescent $\left(\ell_K^2\right)$.   The estimate $\hat{\theta}$ of $\theta$ is the one associated with the beta-coalescent based on the total number of segregating sites.    }%
  \label{tab:mlecod}%
  \centering
  \begin{tabular}{lrrrrrrrrrr}%
    \hline
    data &  $\hat{\alpha}$ &  $\log L(\hat{\alpha})$ &   $\hat{\psi}$  &    $\log L(\hat{\psi})$  &  $\log L(0)$ & $\hat{\theta}$   & $\ell^2_B$   & $\ell^2_K$  \\
    \hline
    Faroe Islands &  $1.28$ & $-82.239$ & $0.06$ & $-89.634$&  $-97.543$ &  $4.577$   & $4.730$ &     $16.441$    \\
    Greenland (G) & $1.70$  & $-21.445$ & $0.03$ & $-21.615$  & $-21.963$  &  $1.409$    &  $2.394$ &  $2.470$  \\
    Iceland (I) & $1.38$ &   $-29.895$  &  $0.02$  & $-32.452$ &  $-40.547$  &  $0.982$   & $3.683$ &  $10.880$ \\
    G + I   &  $1.47$&   $-55.557$  &  $0.01$  & $-60.403$ &   $-68.523$ & $1.179$  & $5.204$ &  $13.247$  \\ 
    \cite{A04} & $1.48$ &  $-71.099$ & $0.005$   & $-81.846$   &  $-90.594$ & $1.373$   & $4.625$ &  $19.448$ \\ 
    \hline
    \end{tabular}   
  \end{table}

  \begin{table}[h!]
    \caption{Akaike's Information Criterion (AIC) for the five cod datasets. By AIC$(0)$ we denote the AIC associated with the Kingman coalescent. }%
    \label{tab:aic}%
    \begin{tabular}{lrrrrr}%
      \hline
      data &  AIC$(\alpha)$ & AIC$(\psi)$ &  AIC$(0)$  \\
      \hline
      Faroe Islands &  $166.478$  &  $181.268$ & $195.086$ \\
      Greenland (G) &  $44.89$    &  $45.23$   & $43.926$  \\
      Iceland (I)   &  $61.79$    &  $66.904$  & $81.094$  \\
      G + I         &  $113.114$  &  $122.806$ & $137.046$ \\
      \cite{A04}    &  $144.198$  &  $165.692$ & $181.188$  \\
      \hline
      \end{tabular}%
  \end{table}%

  The beta-coalescent yields the highest likelihood for all the
  datasets (Table~\ref{tab:mlecod}). The fit to the dataset reported
  in \cite{A04} is, in particular, much better than for the Kingman
  coalescent, as measured by the $\ell^2$-distance.  The superior fit
  of the beta-coalescent is further illustrated in
  Figure~\ref{fig:codalldata}, in which the expected site-frequency
  spectrum of the beta-coalescent and the Kingman coalescent is
  compared to the observed one of \cite{A04}.  A similar graph
  (Figure~\ref{fig:farisdata}) for the Faroe Islands data of
  \cite{SA03} is given in Supporting Information.  Comparing
  \cite{A73}'s Information Criterion (Table~\ref{tab:aic}) between
  models tells the same story (except for  the Greenland data).

\begin{figure}[h!]
  \caption{The folded freq.\ spectrum (white bars) of the data of \cite{A04} along with predictions of the Kingman coalescent (dark-grey bars), and the Beta$(2-\hat{\alpha},\hat{\alpha})$-coalescent  (light-grey bars). The vertical lines represent the standard deviation; obtained for the Beta$(2-\hat{\alpha},\hat{\alpha})$-coalescent  from $10^5$ iterations. The class labelled `11' represents the collated tail of the spectrum, from 11 to 1278/2.    }
  \label{fig:codalldata}%
\includegraphics[width=5in,height=5in]{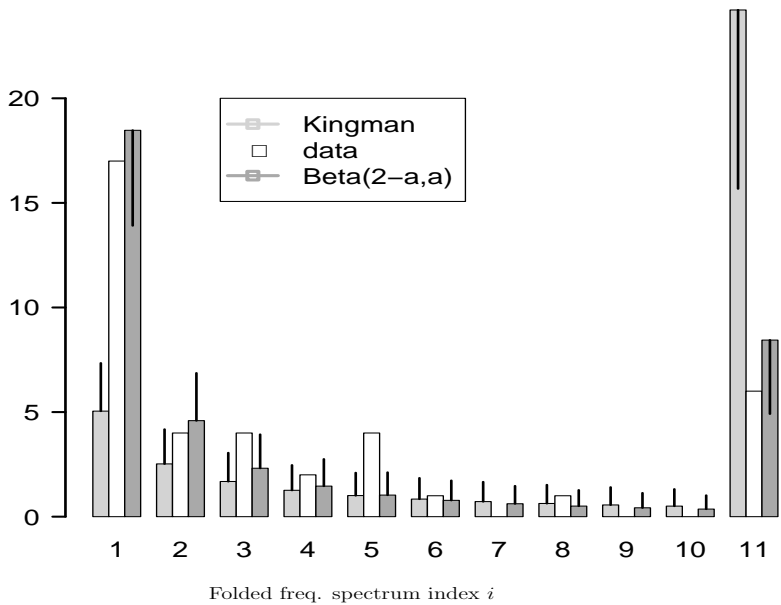}
\end{figure}%

\begin{center}%
{\sc Discussion}
\end{center}%

Inferring coalescence parameters of Lambda-coalescents from large
samples is our main focus.  To this end we derive recursions for the
expected values and covariances of the site-frequency spectrum
associated with a Lambda-coalescent.  The multiple-merger property of
Lambda-coalescents makes obtaining closed-form solutions, as done
by \cite{F95} for the Kingman coalescent, prohibitively hard.  The
value of the recursions in obtaining point-estimates of coalescent
parameters for two specific Lambda-coalescents is then illustrated in
two ways, utilising mtDNA  data on Atlantic cod.

Asymptotic results play a key role in inference.  Knowing for what
regions of the parameter space the asymptotic results approximate the
true values `well enough', in some sense, in other words knowing the
rate of convergence, is also quite important for inference.  We
investigate the rate of convergence for recent asymptotic results, as
sample size $n \to \infty$, of \cite{BBL12} regarding the scaled
site-frequency spectrum associated with Lambda-coalescents, since we
do not have an expression for the error in the asymptotic results.
Our results show that the rate of convergence, in terms of sample
size, is quite slow for low values of the coalescence parameter
$\alpha$.  Thus, exact inference  methods appear preferable.

A simple distance measure is employed on the scaled frequency
spectrum, removing the need to jointly estimate the mutation rate.
The observed spectrum is compared to the expected values for different
values of the coalescence parameters, and the value that minimizes the
distance is reported.  A straightforward pseudo-likelihood approach is
also applied, where the probabilities of mutations to be in different
classes are estimated from the recursions for the expected values of
the spectrum.  These two different approaches give very similar
estimates for the two Lambda-coalescents we consider.  In addition,
they are shown to be capable of distinguishing Lambda-coalescents from
the Kingman coalescent, as well as distinguishing between the two
Lambda-coalescents when sample size is not too small.  Both approaches
also show promise in distinguishing between different
Lambda-coalescents, and they do not require estimates of the mutation
rate.  Our results open up new possibilities to construct test
statistics to test for `neutrality' when the underlying population
model admits large offspring numbers, i.e.\ in highly fecund natural
populations such as Atlantic cod \citep{A04}.  One could, by way of
example, consider linear weighted combinations of the site frequency
spectrum to construct tests for `neutrality' in high fecundity
populations, in the spirit of \cite{A09}.  The present discussion will
not address the statistical power of different statistics one could
construct using our recursions, and their performance under various
scenarios, such as population expansion.  This will be the subject of
future work.  Our aim here is to illustrate the wide applicability of
the recursions we obtain for the mean and variance of the
site-frequency spectrum associated with Lambda-coalescents.

Lambda-coalescents are essentially single-locus models, as they admit
only one merger each time.  Models for multiple loci would naturally
involve Xi-coalescents admitting simultaneous multiple mergers.
Intuitively, one might think that multi-loci inference methods would
be better able to distinguish between different coalescent processes,
and, indeed, the underlying population models, than single-locus
approaches.  A natural question to address is the sample size, as well
as the number of loci, one would recommend in order to identify the
appropriate coalescent process.  Such issues are truly relevant, not
least in part by advances in DNA sequencing technology.  Our partial
response is that our results indicate that one may only need moderate
sample sizes in order to distinguish between Lambda-coalescents.  The
question on the number of loci required is still very much open, not
least because the development of ancestral recombination graphs
admitting simultaneous multiple mergers is only in it's early stages
\citep{EW08,BBE13}.

A key distinguishing feature of the site-frequency spectrum drawn from
a large offspring number population is the excess of singletons
compared to the spectrum one expects in an ordinary Wright-Fisher
population.  When data contains sequencing errors, some authors have
proposed analysing the data without the singletons, leading to
corresponding test statistics which exclude the singletons
\citep{A08}.  At present we do not address the issue of sequencing
errors, but we point out that for many natural high-fecundity
populations, and with advances in DNA sequencing techniques, this may
not be an issue.  However, the issue of sequencing errors remains an
important subject for future analysis.

The role of mtDNA as a suitable genetic marker for studying population
history in general has been subject to some criticism
\citep{BW04,BGG06,B10}.  A single genetic marker may not be sufficient
to infer population history, due to stochasticity in the effects of a
populations' demography on the genome \citep{BW04}.  Natural selection
may be continously acting on (at least) parts of the mitochondrial
genome of Atlantic cod, thus biasing any inference made on supposedly
`neutral' segments of the non-recombining mitochondrial genome
\citep{SH74,G00,DS05}.  In addition, we have not compared the data to
models of population growth, which may mimic the effects of high
fecundity.  However, overfishing of Atlantic cod may have
significantly diminished cod stocks throughout its' range in the
North-Atlantic, which suggests that population expansion may not be an
appropriate model for Atlantic cod.  In addition, as mtDNA is
inherited maternally, and female cods may lay millions of eggs
annually \citep{M67,A04}, our fit with Lambda coalescents might
represent sweepstakes reproduction.  Whatever the main reasons for the
observed patterns of genetic diversity observed (see \cite{A04}, in
particular the discussion on pp.\ 1882 -- 3), we stress that our
purpose is simply to illustrate our inference methods, rather than
reach a firm conclusion on the main mechanism shaping mtDNA diversity
in Atlantic cod.

The two specific Lambda-coalescents we consider will certainly not
(one hopes) be the only Lambda-coalescents ever applied to
highly-fecund natural populations.  Recent work \citep{HM11,HM11b,M11}
introduces specific Lambda-coalescent processes not considered in the
present discussion.  It remains an important exercise to compare all
those new processes to actual data.  In conclusion, we emphasize that
even though one finds that one particular coalescent process fits the
data better than another one, this does not automatically mean that
one has found the correct model.  Biological knowledge and insight
remain key factors in identifying appropriate population models for
natural populations.  By way of example, the haploid population models
underlying the two Lambda-coalescents we focus on -
\eqref{eq:EWoneatom} and \eqref{eq:betameasure} - are very different.
In the model underlying \eqref{eq:betameasure}, the key assumption is
about the distribution of the number of viable (potential) offspring
each individual contributes to the population \citep{schweinsberg03}.
In the model underlying \eqref{eq:EWoneatom}, a single parent
contributes a fixed number of offspring in each reproduction event
\citep{EW06}.  Such assumptions about the population model must form
an integral part of any inference study.

{\small  We thank two referees for insightful comments which improved the presentation.    M.\ Birkner was supported in part by DFG grant BI 1058/2-1.
  J.\ Blath and B.\ Eldon were supported in part by DFG grant BL
  1105/3-1.  }

\begin{appendix}

\section{Appendix}

\setcounter{figure}{0}
\renewcommand{\thefigure}{A\arabic{figure}}
\renewcommand{\thetheorem}{A\arabic{theorem}}%

\subsection{A recursion for $p^{(n)}[k,b]$ (prop.\ \ref{prop:r})}

Assume we start the block counting process in some state $Y_0=n$ and denote its distribution by $\P_n$. Then, the probability that the first jump takes it to state $n' < n$, conditionally on the event that it hits state $k \le n'$ is given by
\begin{align}
\label{conditioning}
\P_n\{Y \mbox{ first jumps from $n$ to } n' \, | \,  Y \mbox{ hits } k \} 
       &= \frac{\P_n \{Y \mbox{ first jumps to } n' \cap Y \mbox{ hits } k \}}
       				{\P_n \{ Y \mbox{ hits } k \}}\notag \\ 
       &= p_{n, n'} \frac{\P_{n'} \{ Y \mbox{ hits } k \}}
       				{\P_n \{ Y \mbox{ hits } k \}} = p_{n, n'} \frac{g(n', k)}{g(n, k)},
\end{align}
due to the Markov property of $Y$. This well-known observation will be useful later.

Decomposing according to 
the first jump of $Y$, starting from $n$, yields the following recursion for  $p^{(n)}[k,b]$.

\begin{prop} 
\label{prop:r}
For $1 < k \le n$, we have
\begin{equation} 
\label{eq:rnk}
{p^{(n)}[k,b]} = 
\sum_{n'=k}^{n-1} p_{n,n'}\frac{g(n',k)}{g(n,k)} 
\bigg( \one_{(b > n-n')} \frac{b-(n-n')}{n'} {p^{(n')}[k,b-(n-n')]}
       + \one_{(b < n')} \frac{n'-b}{n'} {p^{(n')}[k,b]} \bigg), 
\end{equation}
with the obvious boundary conditions ${p^{(n)}[n,b]}=\delta_{1b}$ and 
${p^{(n)}[k,b]} = 0$ if $b > n-(k-1)$.
\end{prop}

The boundary conditions allow us to replace the second indicator
function on the rhs by $\one_{(b \le n-(k-1))}$.   In case of the Kingman coalescent (indicated by an additional superscript-`${(K)}$'), ${p^{(n)}[k,b]}$ becomes
$$
{p^{(n)(K)}[k,b]} =  \one_{(b > 1)}\frac{b - 1 }{n - 1}{p^{(n-1)(K)}[k,b-1]}  +  \one_{(b < n-1)}\frac{n - 1 - b}{n - 1}{p^{(n-1)(K)}[k,b]}
$$
and \cite{F95} obtains a closed-form solution: $ {p^{(n)(K)}[k,b]} =
{\binom{n - b - 1}{k - 2} }/{\binom{n - 1 }{k - 1 } }.  $ A proof of
proposition~\eqref{prop:r} now follows.


\begin{proof}
Conditionally on $Y$ hitting $k$, the first jump can take $Y$ to any $n' \in \{k, \dots, n-1\}$.
The probability of seeing a jump from $n$ to $n'$, conditionally on hitting $k$, has probability
$p_{n,n'}\frac{g(n',k)}{g(n,k)}$, by \eqref{conditioning}. Then, thinking `forwards in time from $n'$ lineages', either the initial $(n-n'+1)$-split occurred to one of the (then necessarily $b-(n-n')\,$)
lineages subtended to the one we are interested in, or it occurs to
one of the (then necessarily $n'-b$) others.  
\end{proof}

\subsection{Proof of Proposition~\ref{prop:Exi}}
\begin{proof}
Let 
$ 
T_k :=
\int_0^\infty \one_{(Y_s = k)} \, ds
$ 
be the length of the time interval
during which there are $k$ lineages (which is possibly $0$), and consider, in an $n$-$\Lambda$-coalescent,
\begin{equation} 
\psi_n(b) := \mbox{expected total length of all branches with $b$ subtended 
leaves}. \notag
\end{equation}
One now obtains
\begin{equation} 
\label{eq:psi} 
\E[\xi_i^{(n)}] = \frac{\theta}{2} \psi_n(i) = 
\frac{\theta}{2}\sum_{k=2}^{n-i+1} {p^{(n)}[k,i]} k \E\big[ T_k \big]
= \frac{\theta}{2}\sum_{k=2}^{n-i+1} {p^{(n)}[k,i]} k g(n,k). 
\end{equation}
\end{proof}

\subsection{ A recursion for  ${p}^{(n)}_{\mathrm{eq}}[k; i,j]$  (prop~\ref{prop:1}) }

\begin{prop}
\label{prop:1}
For $2 \le k < n$ and $i, j \ge 1; i+j \le n-k+2$, we have
\begin{align} 
\label{eq:p.eq.rec}
{p}^{(n)}_{\mathrm{eq}}[k; i,j] & = 
\sum_{m=k}^{n-1} p_{n,m} \frac{g(m,k)}{g(n,k)} 
\Big[ \frac{i-(n-m)}{m}  {p}^{(m)}_{\mathrm{eq}}[k; i-(n-m),j] 
\one_{(i>n-m)} \notag \\
& \hspace{9em} {} 
+ \frac{j-(n-m)}{m}  {p}^{(m)}_{\mathrm{eq}}[k; i,j-(n-m)] 
\one_{(j>n-m)} \notag \\
& \hspace{9em} {} 
+ \frac{m-i-j}{m} {p}^{(m)}_{\mathrm{eq}}[k; i,j] 
\one_{(i+j>m)}  \Big], 
\end{align}
and (as boundary condition)
\begin{align} \label{eq:p.eq.rec.bc}
{p}^{(n)}_{\mathrm{eq}}[n; i,j] = \one_{(i=j=1)}, \quad n \in \N.
\end{align}
\end{prop}

\begin{proof}
  Decompose according to the first jump of $Y$,
  conditionally on hitting $k$. As in the proof of Proposition
  \ref{prop:r}, the probability of jumping to some state $m\in \{k,
  \dots, n-1\}$ is $p_{n,m}\frac{g(m,k)}{g(n,k)}$, by
  \eqref{conditioning}.  Thinking `forwards in time from $m$
  lineages', either the initial $(n-m+1)$-split occured to one of the
  $i-(n-m)$ lineages subtended from the first sampled edge, or to one
  of the $j-(n-m)$ lineages subtended from the second sampled edge, or
  to one of the $m- i-b$ others.
\end{proof}

One can derive a similar recursion for ``labelled'' leaves, where the
order of the sample plays a role. For example, if we distinguish
between individuals in our sample, we obtain for
$p^{(n)}_{\mathrm{eq}}(k; i,j)$, where the round brackets indicate that 
we refer to specific sets of $i$ resp.\ $j$ labelled leaves, the recursion
\begin{align*}
p^{(n)}_{\mathrm{eq}}(k; i,j) & = \sum_{n'=k}^{n-1} p_{n,n'} 
\frac{1}{{n \choose n-n'+1}} \Big[ 
{i \choose n-n'+1} p^{(n')}_{\mathrm{eq}}(k; i-(n-n'),j) \notag \\
& \qquad+ {j \choose n-n'+1} p^{(n')}_{\mathrm{eq}}(k; i,j-(n-n'))  
+ {n-(i+j) \choose n-n'+1} p^{(n')}_{\mathrm{eq}}(k; i,j) \Big]
\end{align*} 
for $n>2$ (noting that $p^{(n)}_{\mathrm{eq}}(k; i,j)=0$ if $i < 1$ or $j < 1$ or 
$i+j \ge n-k+1$).

\subsection{ Recursion for  ${p}^{(n)}_{\mathrm{un}}[k,i ; \ell,j]$   (prop.\ ~\ref{prop:2})  }

\begin{prop}
\label{prop:2}
For $2 \le k < \ell \le n$  we have
\begin{align} 
\label{eq:p.un.rec}
{p}^{(n)}_{\mathrm{un}}[k,i ; \ell,j] 
& = 
\sum_{m=\ell}^{n-1} p_{n,m} \frac{g(m,\ell)}{g(n,\ell)} 
\Big[ \frac{i-(n-m)}{m}  {p}^{(m)}_{\mathrm{un}}[k, i-(n-m); \ell, j] 
\one_{(i>n-m)} \notag \\
& \hspace{9em} {} 
+ \frac{j-(n-m)}{m}  {p}^{(m)}_{\mathrm{un}}[k, i; \ell, j-(n-m)] 
\one_{(j>n-m)} \notag \\
& \hspace{9em} {} 
+ \frac{m-i-j}{m}  {p}^{(m)}_{\mathrm{un}}[k,i; \ell,j] 
\one_{(m>i+j)} \Big] 
\end{align}
with boundary conditions
\begin{align} 
\label{eq:p.un.rec.bc}
{p}^{(n)}_{\mathrm{un}}[k,i ; n,j] = \one_{(j=1)} {p^{(n)}[k,i]} \frac{n-i}{n}, 
\quad  2 \le k < n, \quad 1 \le i <n.
\end{align}
\end{prop}

\begin{proof}
  As for the previous recursions, we decompose according to the first
  jump of $Y$, conditionally on hitting $\ell$. Due to the Markov
  property we do not need to condition on hitting $k$ as well. This
  second conditioning is `hidden' in the ${p}^{(m)}_{\mathrm{un}}$
  terms and in particular in the boundary terms
  \eqref{eq:p.un.rec.bc}.  The rest of the argument is as usual.
\end{proof}

\subsection{A recursion for ${p}^{(n)}_{\mathrm{ne}}[k,i ; \ell ,j]$ (prop.\  \ref{prop:3})  }

\begin{prop}
\label{prop:3}
For $2 \le k < \ell \le n$ we have
\begin{align} 
\label{eq:p.ne.rec}
{p}^{(n)}_{\mathrm{ne}}[k,i ; \ell ,j] & = 
\sum_{m=\ell}^{n-1} p_{n,m} \frac{g(m,\ell)}{g(n, \ell)} 
\Big[ \frac{i-j-(n-m)}{m}  {p}^{(m)}_{\mathrm{ne}}[k, i-(n-m); \ell, j] 
\one_{(i-j>n-m)} \notag \\
& \hspace{9em} {} 
+ \frac{j-(n-m)}{m} {p}^{(m)}_{\mathrm{ne}}[k, i-(n-m); \ell, j-(n-m)]
\one_{(j>n-m)} \notag \\
& \hspace{9em} {} 
+ \frac{m-i}{m} {p}^{(m)}_{\mathrm{ne}}[k,i; \ell,j] 
\one_{(m>i} \Big] 
\end{align}
with boundary conditions
\begin{align} \label{eq:p.ne.rec.bc}
{p}^{(n)}_{\mathrm{ne}}(k,i ; n,j) = \one_{(j=1)} {p^{(n)}[k,i]} \frac{i}{n}, 
\quad  2 \le k < n, 1 \le i < n.
\end{align}
\end{prop}

The proof follows from arguments similar to the previous proofs.

The recursions for  $p^{(n)}_{\mathrm{eq}}$ \eqref{eq:p.eq.rec}, $p^{(n)}_{\mathrm{un}}$ \eqref{eq:p.un.rec}, and 
$p^{(n)}_{\mathrm{ne}}$   \eqref{eq:p.ne.rec}, are (strictly) recursive in $n$ and can be solved numerically in a straightforward way, starting
from their respective boundary conditions \eqref{eq:p.eq.rec.bc},
\eqref{eq:p.un.rec.bc}, \eqref{eq:p.ne.rec.bc}. 
%
For numerical checking of implementations it may be useful 
to note that by definition,  
\begin{align}
\sum_{i,j} p^{(n)}_{\mathrm{eq}}(k; i,j) &= 1 \quad 
\text{ for all } 2\le k \le n, \\
\sum_{i,j} \big( p^{(n)}_{\mathrm{un}}(k,i ; \ell ,j) 
+ p^{(n)}_{\mathrm{ne}}(k,i ; \ell,j) \big) &= 1 \quad 
\text{ for all } 2\le k<\ell \le n.
\end{align}

\subsection{Proof of Theorem~\ref{thm:cov} }
\label{sec:proofcov}%

Consider  an  $n$-$\Lambda$-coalescent, and let, for $2 \le k \le n$, $\ell \in
[k]$,
\begin{align}
\label{eq:Lklndef}
L^{(n)}_{k,\ell} = 
\begin{cases} 
\# \, \text{leaves subtended to the $\ell$-th edge while $k$ blocks 
(in $n$-coalescent)},\\
0 \;\; \text{if $k$ blocks are never realised}.
\end{cases}
\end{align} 
Now we think of labelled lines, as opposed to the unlabelled lines in
the previous recursions. By symmetry, $\P\{L^{(n)}_{k,\ell}=i\}$ does
not depend on $\ell \in [k] := \{1, \ldots, k\}$.  In fact
\begin{align} 
\label{1}
\P\{L^{(n)}_{k,\ell}=i\} = {p^{(n)}[k,i]} \frac{g(n,k)}{g(k,k)}, \quad \mbox{ for } 2 \le k \le n, \: \ell \in [k].
\end{align}
The factor $g(n,k)/g(k,k)$ is the probability that the block counting 
process $Y$, starting from $n$ lines, ever hits the state of $k$ lines, 
and thus removes the conditioning from the definition of ${p^{(n)}[k, \cdot]}$.
Similarly, for $2 \le k \le n, \: \ell, \ell' \in [k], \ell \neq \ell',$
\begin{equation}
\label{2}
\P\{L^{(n)}_{k,\ell}=i, L^{(n)}_{k,\ell'}=j\}  = 
\P\{L^{(n)}_{k,1}=i, L^{(n)}_{k,2}=j\} = 
{p}^{(n)}_{\mathrm{eq}}[k; i,j]\frac{g(n,k)}{g(k,k)}, 
\end{equation}
and, for $2 \le k<k' \le n, \: \ell \in [k], \ell' \in [k']$,
\begin{equation}
\label{3}
\P\{L^{(n)}_{k,\ell}=i, L^{(n)}_{k',\ell'}=j\}  = 
 \P\{L^{(n)}_{k,1}=i, L^{(n)}_{k',1}=j\} 
  = \big( p^{(n)}_{\mathrm{un}}(k,i ; k',j) 
+ {p}^{(n)}_{\mathrm{ne}}[k,i ; k',j] \big)
\frac{g(n,k')}{g(k',k')} \frac{g(k',k)}{g(k,k)}.
\end{equation}

The random variables $\xi_i^{(n)}$ counting the number of (derived)
mutations in $i$ copies (for $1 \le i < n)$ can be expressed in terms
of $L_{k,\ell}^{(n)}$ as follows.  Write $M^{(n)}_{k,\ell}$ for the
number of mutations that occur on the $\ell$-th edge while there are
$k$ blocks in our $n$-$\Lambda$-coalescent, with $M^{(n)}_{k,\ell} \equiv 0$   if $k$ blocks are never
realised.  Then
\begin{align} 
\xi^{(n)}_i = \sum_{k=2}^n \sum_{\ell=1}^k \one_{\left(L^{(n)}_{k,\ell}=i \right)} M^{(n)}_{k,\ell}.
\end{align}

\begin{proof}
  As in \eqref{eq:Lklndef}, let $L^{(n)}_{k,\ell}$ denote the random number of leaves - out
  of $n$ - subtended to block $\ell$ when there are $k$ blocks, with
  $L^{(n)}_{k,\ell} = 0$ if $k$ blocks are never realised.

  Distinguishing cases as before, we obtain
\begin{align*} 
\E\big[\xi^{(n)}_i \xi^{(n)}_j \big] &=  \,  
\sum_{k=2}^n \sum_{\ell=1}^k \sum_{k'=2}^n \sum_{\ell'=1}^{k'} 
\P\{L^{(n)}_{k,\ell}=i, L^{(n)}_{k',\ell'}=j\} 
\E\big[ M^{(n)}_{k,\ell} M^{(n)}_{k',\ell'} \big] \notag \\
&= \, \sum_{k=2}^n k(k-1) \P\{L^{(n)}_{k,1}=i, L^{(n)}_{k,2}=j\} 
\E\big[ M^{(n)}_{k,1} M^{(n)}_{k,2} \big] \\
 & \quad\:\;\; + \one_{(i=j)} \sum_{k=2}^n k \P\{L^{(n)}_{k,1}=i\}
\E\big[ (M^{(n)}_{k,1})^2 \big] \\
 & \quad\;\; + 2\sum_{k=3}^n \sum_{k'=2}^{k-1} k \, k' \, 
\P\{L^{(n)}_{k,1}=i, L^{(n)}_{k',1}=j\} \E\big[ M^{(n)}_{k,1} M^{(n)}_{k',1} \big]
\end{align*}
The three expected values on the right-hand side can now be expressed in terms of the total jump rates of our block-counting process $Y$ away from states $k$ resp.\ $k'$. Indeed, for the first 
expectation on the {\em rhs} note that the time while there are $k$ blocks is the same for both edges 1 and 2, but mutations are collected {\em independently}. Thus, the expected number of mutations can  be computed as the second moment of an exponential random variable:
\begin{align*}
\E\big[ M^{(n)}_{k,1} M^{(n)}_{k,2} \big] &= \int_0^\infty 
\E\big[ M^{(n)}_{k,1} M^{(n)}_{k,2} \, \big| \, \mbox{$k$ blocks for $t$ time units}\big] 
(-q_{kk}) e^{-(-q_{kk})t}\, dt\\
&= \int_0^\infty 
\Big(\E\big[ M^{(n)}_{k,1} \, \big| \, \mbox{$k$ blocks for $t$ time units}\big] \Big)^2
(-q_{kk}) e^{-(-q_{kk})t}\, dt\\
&=\int_0^\infty \Big(\frac{\theta}{2} t\Big)^2 
(-q_{kk}) e^{-(-q_{kk})t}\, dt \\
&= \frac{\theta^2}{4} \frac{2}{(-q_{kk})^2}.
\end{align*}
The second expectation is given by the second moment of a mixture of a Poisson random variable with an exponential distribution:
\begin{align*}
\E\big[ \left(M^{(n)}_{k,1}\right)^2 \big] &= \int_0^\infty 
\E\big[ \left(M^{(n)}_{k,1}\right)^2 \, \big| \, \mbox{$k$ blocks for $t$ time units}\big] 
(-q_{kk}) e^{-(-q_{kk})t}\, dt\\
&=\int_0^\infty \Big[ \sum_{k=0}^\infty m^2 \frac{\frac{\theta}{2}t}{m!} e^{-\frac{\theta}{2}t }\Big](-q_{kk}) e^{-(-q_{kk})t}\, dt\\
&= \int_0^\infty \Big[\Big(\frac{\theta}{2}t\Big)^2 + \frac{\theta}{2}t\Big]
(-q_{kk}) e^{-(-q_{kk})t}\, dt\\
&=  \frac{\theta}{2} \frac{1}{-q_{kk}} + 
\frac{\theta^2}{4}\frac{2}{(-q_{kk})^2}.
\end{align*}
The last expectation is  (due to independence obtained from the memoryless-property of coalescent jump times) just the product of the two individual expected values, i.e.
$$
\E\big[ M^{(n)}_{k,1} M^{(n)}_{k',1} \big] = \frac{1}{(-q_{kk})(-q_{k'k'})}.
$$
Altogether, we arrive at
\begin{align*} 
\E\big[\xi^{(n)}_i \xi^{(n)}_j \big] &=  \,
\frac{\theta^2}{4} \sum_{k=2}^n k(k-1) \P\{L^{(n)}_{k,1}=i, L^{(n)}_{k,2}=j\} 
\frac{2}{(-q_{kk})^2}
 \\
& \quad\,\;\; + \one_{(i=j)} \sum_{k=2}^n k \P\{L^{(n)}_{k,1}=i\} 
\Big( \frac{\theta}{2} \frac{1}{-q_{kk}} + 
\frac{\theta^2}{4}\frac{2}{(-q_{kk})^2} \Big)\\
& \quad\;\; + 
2 \frac{\theta^2}{4} \sum_{k=3}^n \sum_{k'=2}^{k-1} k \, k' \, 
\P\{L^{(n)}_{k,1}=i, L^{(n)}_{k',1}=j\} \frac{1}{(-q_{kk})(-q_{k'k'})},
\end{align*}
from which the result follows with the help of \eqref{1}, \eqref{2} and \eqref{3}.
\end{proof}

Let $B_i^{(n)}$ denote the total length of branches subtending $i$ of
$n$ leaves.   The covariance of $B_i^{(n)}$ and $B_j^{(n)}$ now follows immediately from previous results since 
\begin{equation}%
  \label{eq:EBiBj}
  \begin{split}%
    \EE{B_i^{(n)}B_j^{(n)}} & =  \sum_{k=2}^n k(k-1) \P\{L^{(n)}_{k,1}=i, L^{(n)}_{k,2}=j\} \frac{2}{q_{kk}^2} \\
    & +  \one_{(i=j)} \sum_{k=2}^n k \P\{L^{(n)}_{k,1}=i\} \frac{2}{q_{kk}^2} \\
    & +  \sum_{k=3}^n \sum_{k'=2}^{k-1} k \, k' \, \P\{L^{(n)}_{k,1}=i, L^{(n)}_{k',1}=j\} \frac{1}{q_{kk}q_{k'k'}}.\\
    \end{split}%
\end{equation}%

\end{appendix}

\clearpage
\pagebreak
\newpage

\clearpage
\pagebreak
\newpage


\addtocounter{section}{0}
\setcounter{figure}{0}
\renewcommand{\thefigure}{S\arabic{figure}}
\renewcommand{\thetable}{S\arabic{table}}%
\renewcommand{\thesection}{S\arabic{section}}%

\begin{center}%
{\Large  \bf Supporting Information}
\end{center}%
\begin{center}%
  {\Large \bf Statistical properties of the site-frequency spectrum associated with Lambda-coalescents}
\end{center}%
\begin{center}%
  {\bf  Matthias Birkner, Jochen Blath, Bjarki Eldon}
  \end{center}%

\clearpage
\pagebreak
\newpage

\addtocounter{section}{-1}
\section{Fu's covariance formulas for the Kingman case}
\label{ssn:Fu}
\cite{F95} obtains expressions for the expected value $\EE{\xi^{(n)}_i}$,
variance $\Var{\xi^{(n)}_i}$, and covariance \Cov{\xi^{(n)}_i}{\xi^{(n)}_j} of the
number $\xi^{(n)}_i$ of mutations in $i$ copies in a sample of size $n$
whose genealogy {\em is governed by the Kingman coalescent}.  For ease of reference
we now cite \cite{F95}'s results.   Let
$$
a_n :=  1 + \frac{1}{2} + \cdots  +  \frac{1}{n - 1}
$$
for $n \ge 2$ with $a_1 := 0$,  and
$$
\beta_n(i) := \frac{2n}{(n - i + 1)(n - i)}(a_{n + 1} - a_i)  -  \frac{2}{n - i}.
$$
Define the quantities
\begin{displaymath}%
\sigma_{ii} = \begin{cases}\beta_{n}(i + 1), & \mathrm{if} \quad i < n/2, \\
  2\frac{a_n - a_i }{n - i} - i^{-2}, & \mathrm{if} \quad i = n/2, \\
 \beta_n(i)  -  i^{-2}, & \mathrm{if} \quad i > n/2,\\\end{cases}
\end{displaymath}%
and, for $i < j$,
\begin{displaymath}%
\sigma_{ij} = \begin{cases} \frac{\beta_{n}(j + 1) - \beta_n(j)  }{2}, & \mathrm{if} \quad i + j < n, \\
  \frac{a_n - a_i }{n - i}  +  \frac{a_n - a_j }{n - j} -  \frac{\beta_n(i + 1) + \beta_n(j) }{2 }  -    \frac{1}{ij}, & \mathrm{if} \quad i + j = n, \\
 \frac{\beta_n(i)  -  \beta_n(i + 1) }{2}  -  \frac{1}{ij},&  \mathrm{if} \quad i + j > n. \\\end{cases}
\end{displaymath}%

\begin{theorem}[Fu 1995]
\label{thm:Fu}
For the unfolded site frequency spectrum of a $n$-sample with mutation rate $\theta/2$ and genealogy given by a Kingman-coalescent, 
\begin{displaymath}%
  \begin{split}%
    \EE{\xi^{(n)}_i} & = \frac{\theta}{i}, \\
    \Var{\xi^{(n)}_i} & = \frac{\theta}{i}  +  \sigma_{ii}\theta^2, \\
    \Cov{\xi^{(n)}_i}{\xi^{(n)}_j} & = \sigma_{ij}\theta^2,
    \end{split}%
  \end{displaymath}%
  for $i,j=1, \dots, n-1$ with $i \neq j$.
\end{theorem}

For small $j$ and large $n$, the covariances turn negative, i.e.\ $\sigma_{ij}<0$.   

\section{Transition rates in special cases}

In this section we specify transition rates associated with
Lambda-coalescents when the measure $\Lambda$ assumes specific forms.

In the case when $\Lambda=\mathrm{Beta}(a,b)$ ($a, b>0$) with density
\begin{equation} 
\frac{\Gamma(a+b)}{\Gamma(a) \Gamma(b)} x^{a-1} (1-x)^{b-1} 
\one_{(0,1)}(x), 
\end{equation}
the $q_{ij}$ can be computed a little more explicitly: 
\begin{eqnarray*}
\lambda_{n,k} & = & \frac{\Gamma(a+b)}{\Gamma(a) \Gamma(b)} 
          \int_0^1 x^{k+a-3} (1-x)^{n-k+b-1} \, dx 
        = \frac{\Gamma(a+b)}{\Gamma(a) \Gamma(b)}  
          \frac{\Gamma(k+a-2) \Gamma(n-k+b)}{\Gamma(n-2+a+b)} \\ 
& = & \frac{(a)_{k-2} (b)_{n-k}}{(a+b)_{n-2}}, 
\end{eqnarray*} 
where $(x)_i = x(x+1)\cdots(x+i-1)$, $(x)_0=1$, and we 
used $\Gamma(x+1)=x \Gamma(x)$. Thus 
\begin{eqnarray*} 
q_{ij} & = & {i \choose i-j+1} \lambda_{i, i-j+1} 
= \frac{i!}{(i-j+1)! (j-1)!} \frac{(a)_{i-j-1} (b)_{j-1}}{(a+b)_{i-2}}.
\end{eqnarray*} 
In the case of the coalescent \eqref{eq:EWoneatom}, we obtain, for $2 \le k \le n$,
\begin{align*}
\lambda_{n,k} &= \int_0^1 x^{k+2}(1-x)^{n-k}\, \Big(\frac{2}{2+\psi^2}\delta_0(dx)+ \frac{\psi^2}{2+\psi^2}\delta_\psi(dx)\Big)\\
&= \frac{2}{2+\psi^2} \one_{(k=2)} + \frac{\psi^2}{2+\psi^2}\Big[ \psi^{k+2}(1-\psi)^{n-k}\Big].
\end{align*}
In particular, for $\psi=0$, one has $\lambda_{n,k} =\one_{(k=2)}$, and for $\psi=1$, this gives $\lambda_{n,k} =\frac 13 \one_{(k=n)}$.





\clearpage
\pagebreak
\newpage

\begin{figure}[h!]
  \caption{The results of \cite{BBL12}  (Thm.\ \ref{thm:bbl}) regarding convergence of the scaled unfolded frequency spectrum   compared to  simulations with no.\ of leaves $n = 10000$ and varying over $\alpha$ and $\theta$ as shown.       }
  \label{fig:bblfitnbig1}%
  \vspace{-1cm}\hspace{-1.1in}%
\includegraphics[width=8.5in,height=9in]{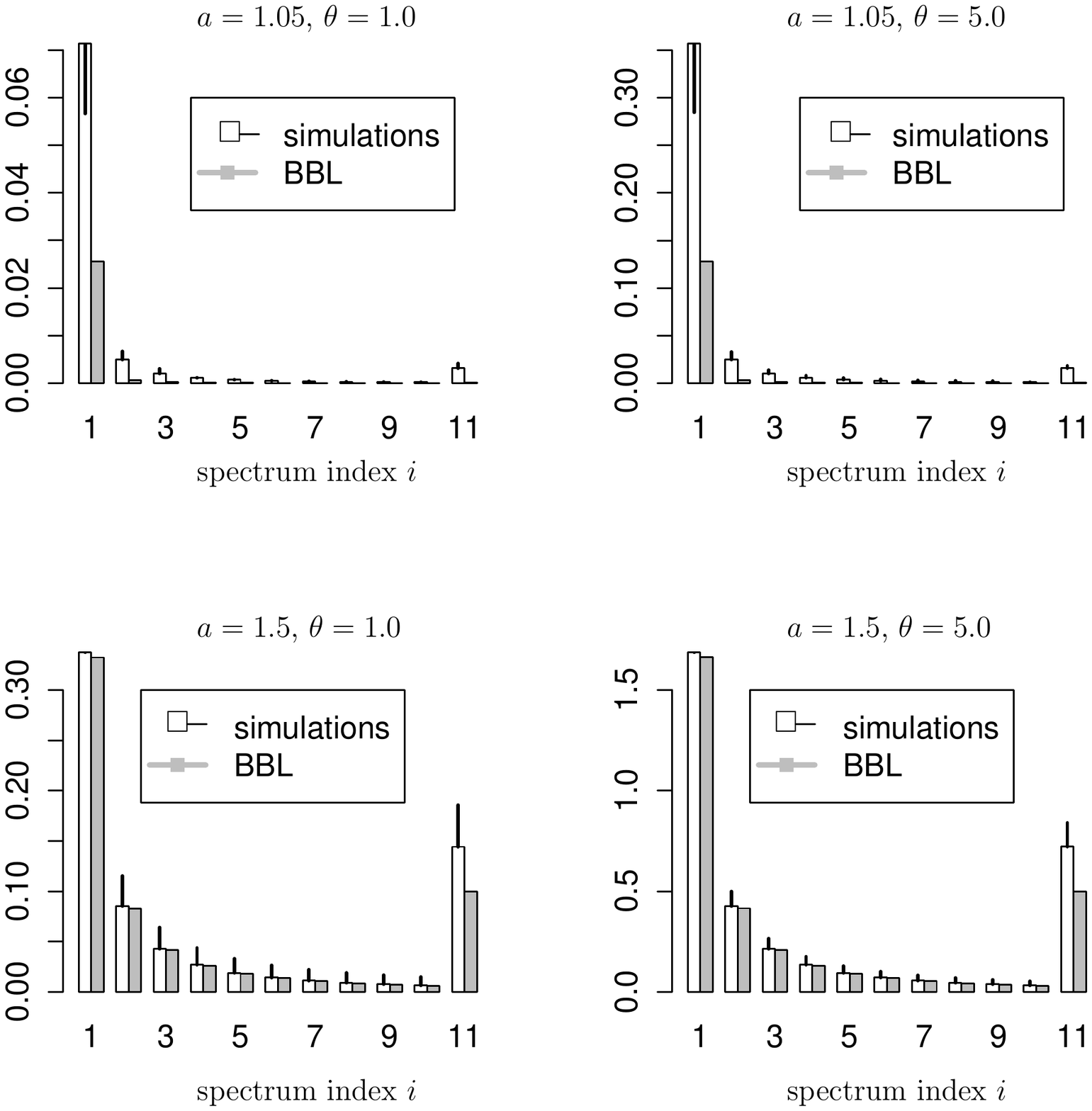}
\vspace{-1cm}
\end{figure}%

\clearpage
\pagebreak
\newpage

\begin{figure}[h!]
  \caption{The results of \cite{BBL12} (Thm.\ \ref{thm:bbl}) regarding convergence of the scaled unfolded frequency spectrum  compared to  simulations with no.\ of leaves $n = 10000$ and varying over $\alpha$ and $\theta$ as shown.   }
  \label{fig:bblfitnbig2}%
  \vspace{-1cm}\hspace{-1in}%
\includegraphics[width=8.5in,height=9in]{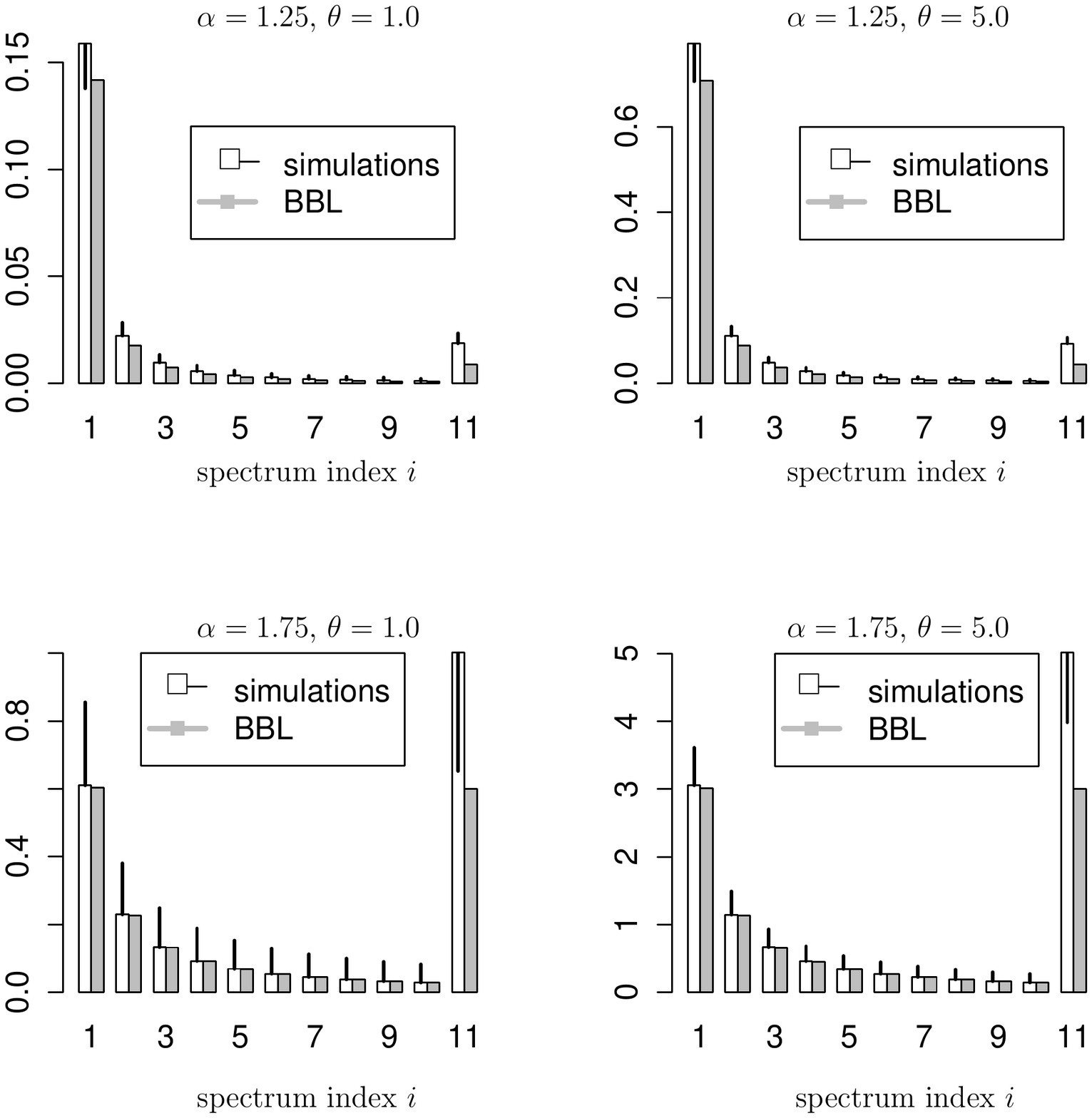}
\vspace{-1cm}
\end{figure}%

\clearpage
\pagebreak
\newpage

\begin{figure}[h!]
  \caption{The results of \cite{BBL12}  (Thm.\ \ref{thm:bbl}) regarding convergence of the scaled unfolded frequency spectrum  (grey bars) compared to  simulations (white bars) varying over $\alpha$ as shown, with $\theta = 1$.     Panels in left column are for sample size  $n=100$, panels in right column are for sample size $n=250$.     }
  \label{fig:bblfitnsmall}%
  \vspace{-1cm}\hspace{-1.1in}
\includegraphics[width=8.5in,height=9in]{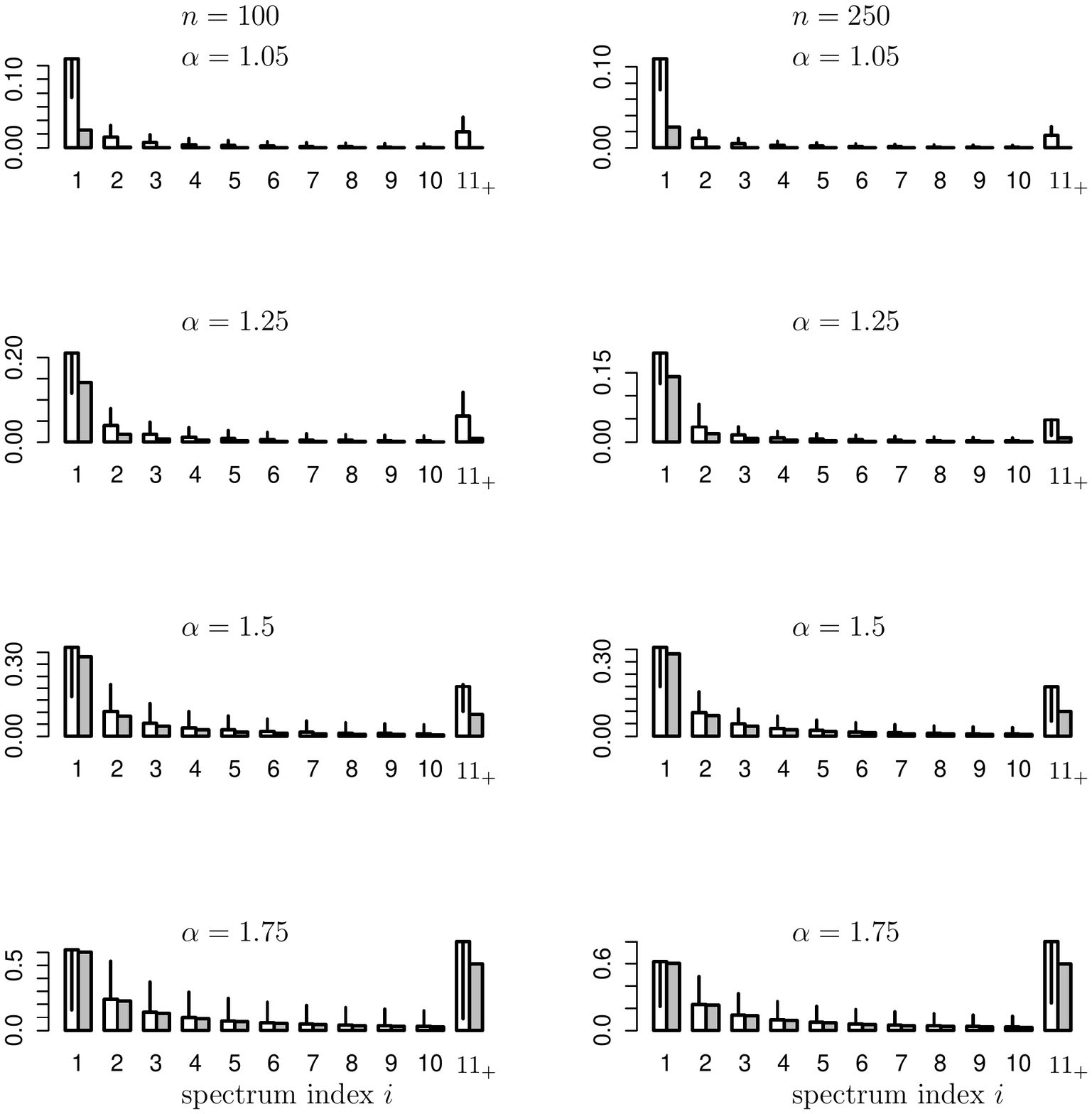}
\vspace{-1cm}
\end{figure}%

\clearpage
\pagebreak
\newpage

\begin{figure}[H!]
  \caption{The folded freq.\ spectrum (white bars) of the data of \cite{SA03} along with predictions of the Kingman coalescent (dark-grey bars), and the Beta$(2-\hat{\alpha},\hat{\alpha})$-coalescent  (light-grey bars). The vertical lines represent the standard deviation; obtained for the Beta$(2-\hat{\alpha},\hat{\alpha})$-coalescent  from $10^5$ iterations. The class labelled `7' represents the collated tail of the spectrum, from 11 to 74/2.    }
  \label{fig:farisdata}%
\includegraphics[width=7in,height=7in]{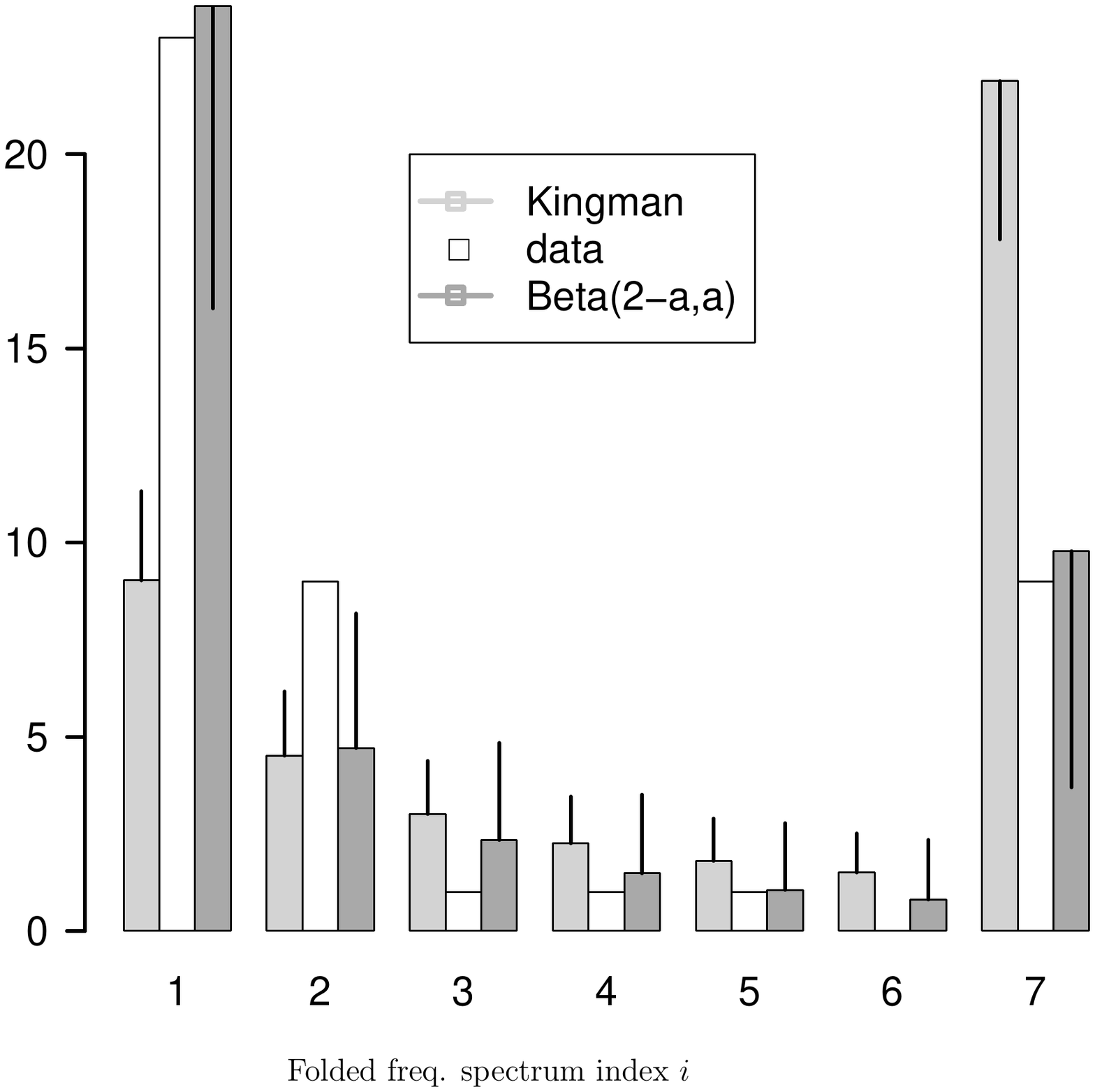}
\end{figure}%

\clearpage
\pagebreak
\newpage


\begin{figure}[H!]
  \caption{ Let $L_i^{(n)}$ denote the random length of branches
    subtending $i$ of $n$ leaves, and let $L^{(n)}$ denote the total
    length of the genealogy.  The ratio $p_i :=
    \EE{L_i^{(n)}}/\EE{L^{(n)}}$ then corresponds to $\varphi_n(i)$.
    Define $R_i := L_i/L$, and denote by $\overline{R_i}$ an estimate
    of $\EE{R_i}$, the correct quantity we approximate with
    $\varphi_n(b)$.  Comparison of the ratios $p_i :=
    \mathbb{E}[L_i]/\mathbb{E}[L]$ and $\overline{R_i}$ for the
    Beta$(2-\alpha,\alpha)$-coalescent with no.\ of leaves $n$ and
    $\alpha = \textrm{\tt a}$ varying as shown in the legend. Entries
    are $\left|p_i - \overline{R_i} \right|/\overline{R_i}$, in which
    $\overline{R_i}$ are estimates of $\mathbb{E}[R_i]$ obtained from
    $10^5$ iterations.  }
  \label{fig:errfit0}%
\includegraphics[width=6in,height=6.6in]{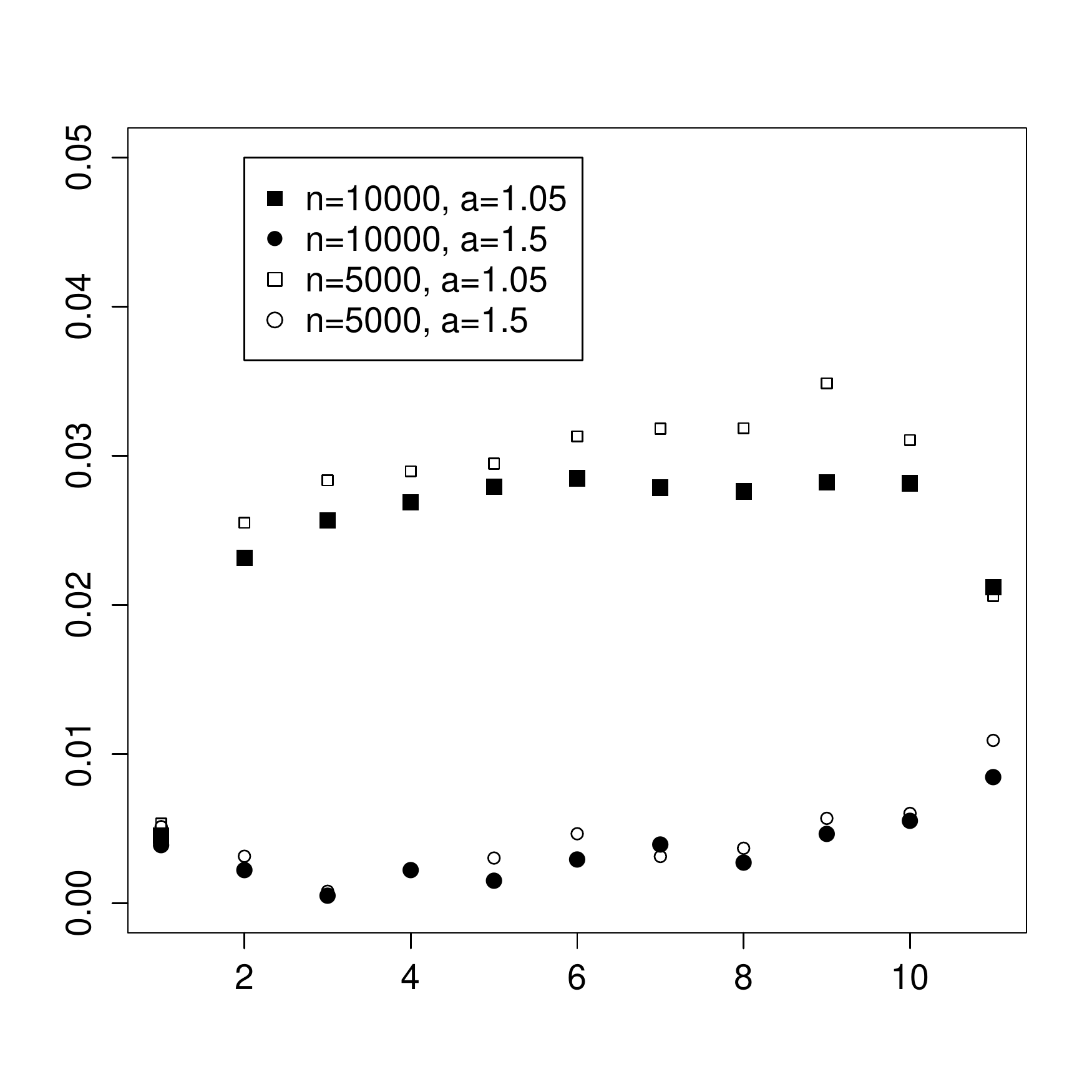}
\end{figure}%

\clearpage
\pagebreak
\newpage

\begin{figure}[H!]
  \caption{ Let $L_i^{(n)}$ denote the random length of branches
    subtending $i$ of $n$ leaves, and let $L^{(n)}$ denote the total
    length of the genealogy.  The ratio $p_i :=
    \EE{L_i^{(n)}}/\EE{L^{(n)}}$ then corresponds to $\varphi_n(i)$.
    Define $R_i := L_i/L$, and denote by $\overline{R_i}$ an estimate
    of $\EE{R_i}$, the correct quantity we approximate with
    $\varphi_n(b)$.  Comparison of the ratios $p_i :=
    \mathbb{E}[L_i]/\mathbb{E}[L]$ and $\overline{R_i}$ for the
    Beta$(2-\alpha,\alpha)$-coalescent with no.\ of leaves $n$ and
    $\alpha = \textrm{\tt a}$ varying as shown in the legend. Entries
    are $\left|p_i - \overline{R_i} \right|/\overline{R_i}$, in which
    $\overline{R_i}$ are estimates of $\mathbb{E}[R_i]$ obtained from
    $10^5$ iterations.  }
  \label{fig:errfit1}%
\includegraphics[width=6in,height=6.6in]{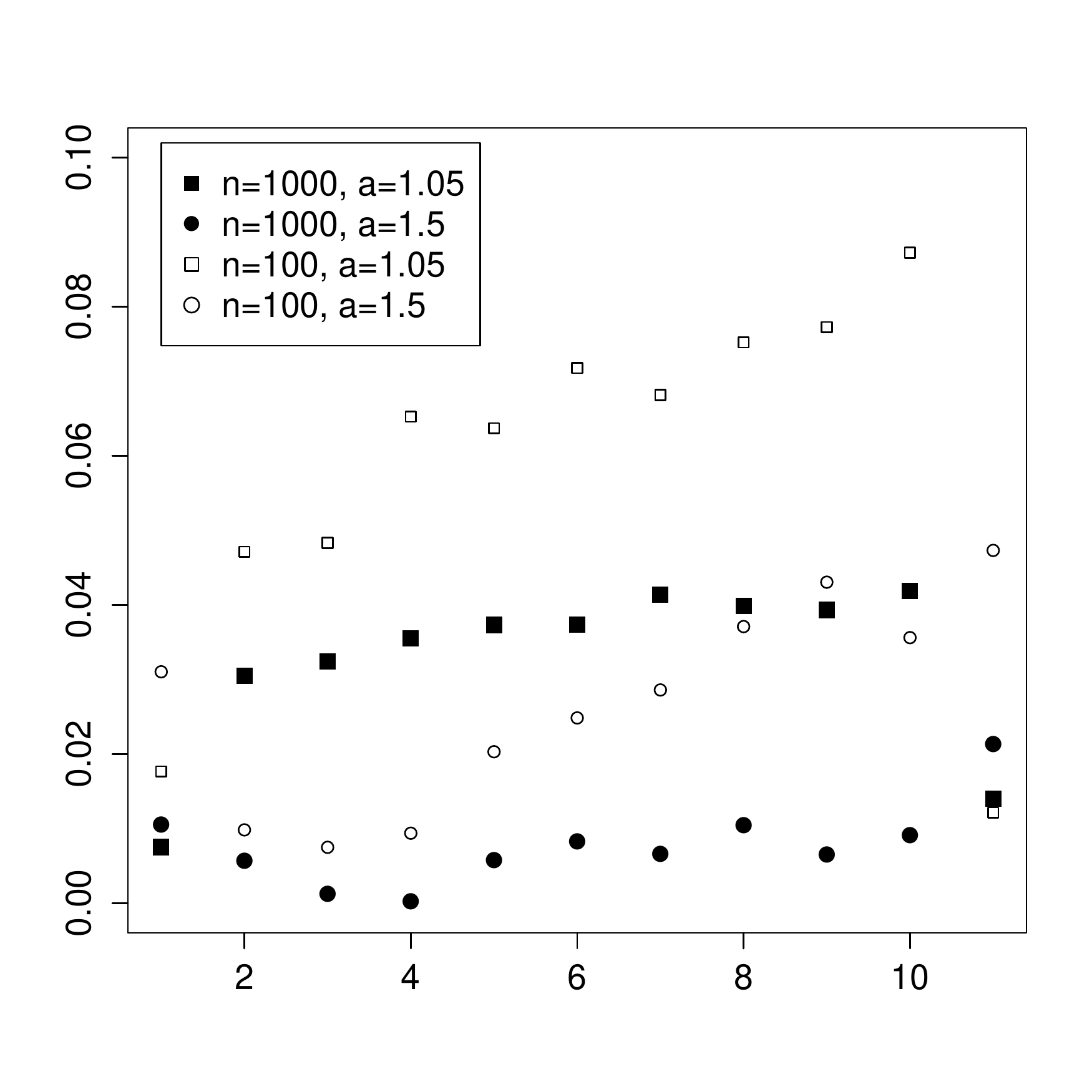}
\end{figure}%

\clearpage
\pagebreak
\newpage


Approximate bayesian computation
\citep[ABC;][]{PSPF99,BZB02,BCMR09,MPRR11} seems well suited to
distinguish Lambda-coalescent processes from the Kingman coalescent,
and to obtain credible intervals for the coalescence parameters. By
way of example, \cite{ZT12} apply ABC on the frequency spectrum to
distinguish between possible demographic scenarios in relation to seed
banks.  Although we do not apply ABC for now, we do consider the
distribution of the $\ell^2$-norm when one compares the observed
spectrum with one simulated from a Lambda-coalescent.  Thus, let
$$
\boldsymbol{\xi}^{(*,n)} := \left(\xi_1^{(*,n)}, \ldots, \xi_{n-1}^{(*,n)} \right) 
$$ 
denote  simulated spectrum obtained under a particular coalescent process.   The $\ell^2$-norm is now
$$
\ell^2 =  \sqrt{\sum_i \left( \xi_i^{n}  -  \xi_i^{(*,n)}  \right)^2 }
$$
where $\xi_i^{(n)}$ denotes the observed spectrum.
Figures~\ref{fig:farisltwodist}--\ref{fig:iceltwodist} in
 show the distribution of the $\ell^2$-norm for
the different Atlantic cod datasets based on $1000$ iterations for
each parameter value.  Each vertical line denotes the interval spanned
by the quantiles $5\%$ and $95\%$, while the closed square on each
line denotes the median.  By way of example, one would clearly reject
the Kingman coalescent in favor of the beta-coalescent for the
\cite{A04} dataset (Figure~\ref{fig:allltwodist}).  In
Figures~\ref{fig:betaone} and \ref{fig:betatwo} the expected values
for the beta-coalescent for $\alpha = 1.5$ and $\alpha = 1.05$, resp.,
replaces the data; i.e.\ we consider the distribution of the
$\ell^2$-norm if the data were exactly the expected site-frequency
spectrum associated with the beta-coalescent.  The Figures indicate
that an analysis of the distribution of the $\ell^2$-norm is a
promising way to distinguish Lambda-coalescents from the Kingman
coalescent. 


\begin{figure}[H!]
  \caption{The $5\%$ and $95\%$ quantiles (lines) and the median (squares) of the $\ell^2$-distance  between simulated data for each parameter value ($\alpha$ or $\psi$) and the Faroe Islands cod dataset of \cite{SA03}.  The quantiles for each parameter value are based on $1000$ iterations.   The $\alpha$ value $1.435$ yields  the minimum $95\%$ quantile of  $17.177$, while the value $0.06$ of $\psi$ yields  the minimum $95\%$ quantile of $19.081$.       }
  \label{fig:farisltwodist}%
\hspace{-1in}\includegraphics[width=8in,height=8in]{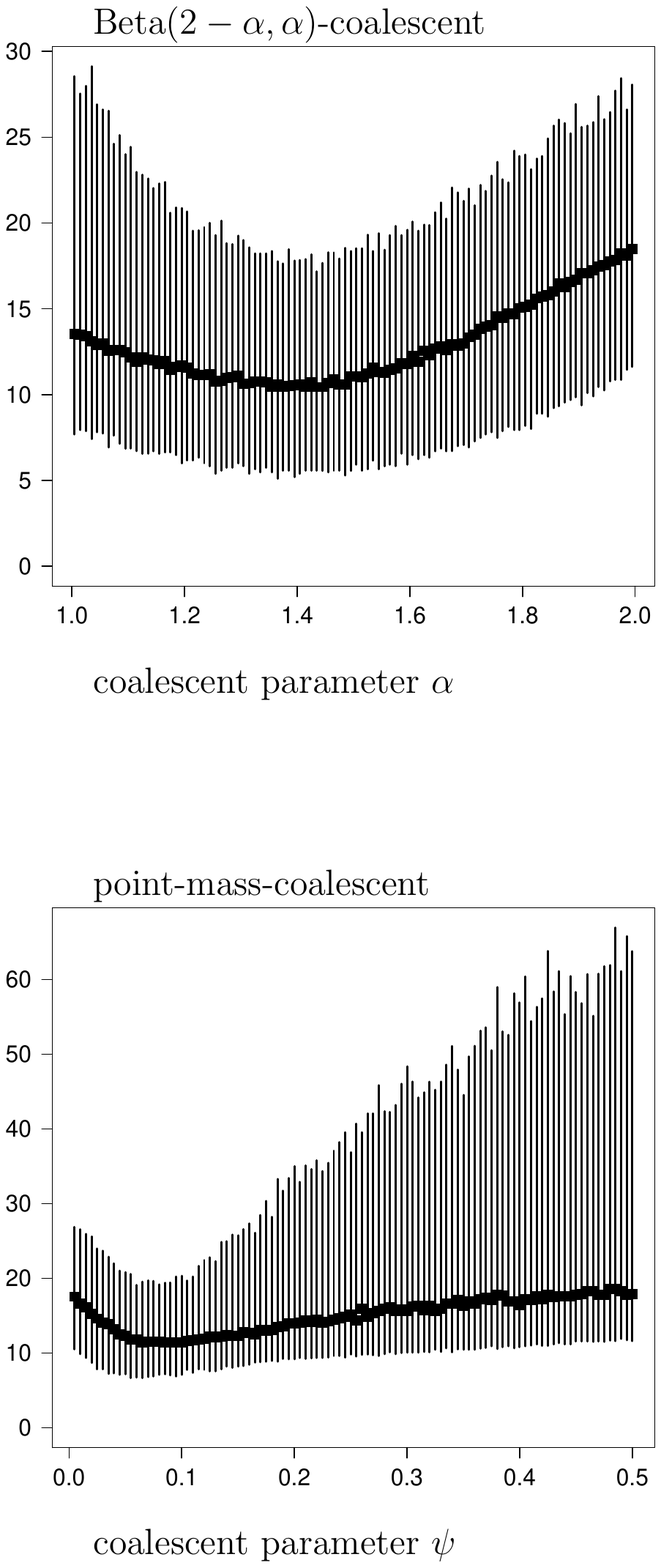}
\end{figure}%

\clearpage
\pagebreak
\newpage

\begin{figure}[H!]
  \caption{The $5\%$ and $95\%$ quantiles (lines) and the median (squares) of the $\ell^2$-distance  between simulated data for each parameter value ($\alpha$ or $\psi$) and the Greenland cod dataset of \cite{A00}.  The quantiles for each parameter value are based on $1000$ iterations.   The $\alpha$ value $1.665$ yields  the minimum $95\%$ quantile of  $5.661$, while the value $0.025$ of $\psi$ yields  the minimum $95\%$ quantile of $5.839$.       }
  \label{fig:greenltwodist}%
\hspace{-1in}\includegraphics[width=8in,height=8in]{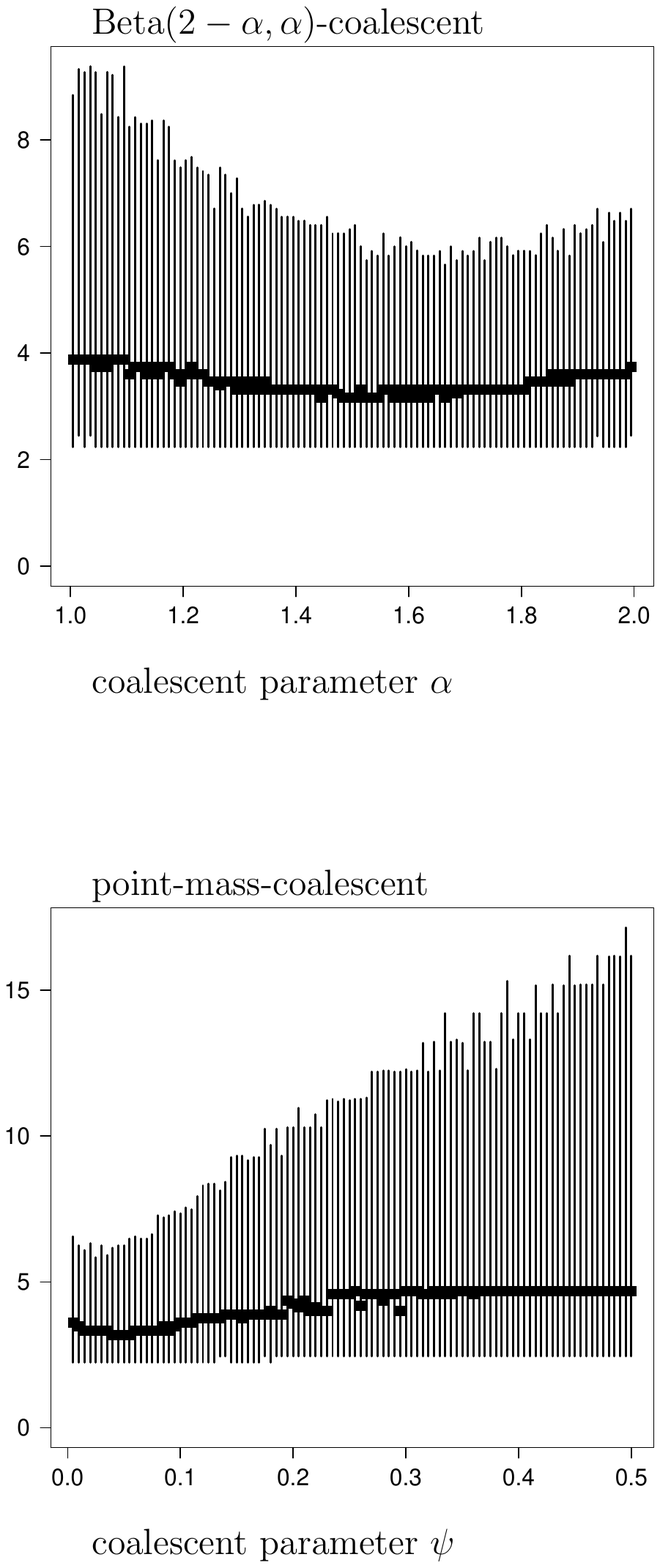}
\end{figure}%

\clearpage
\pagebreak
\newpage

\begin{figure}[H!]
  \caption{The $5\%$ and $95\%$ quantiles (lines) and the median (squares) of the $\ell^2$-distance  between simulated data for each parameter value ($\alpha$ or $\psi$) and the Iceland cod dataset of \cite{A00}.  The quantiles for each parameter value are based on $1000$ iterations.   The $\alpha$ value $1.405$ yields  the minimum $95\%$ quantile of  $9.168$, while the value $0.02$ of $\psi$ yields  the minimum $95\%$ quantile of $9.806$.       }
  \label{fig:iceltwodist}%
\hspace{-1in}\includegraphics[width=8in,height=8in]{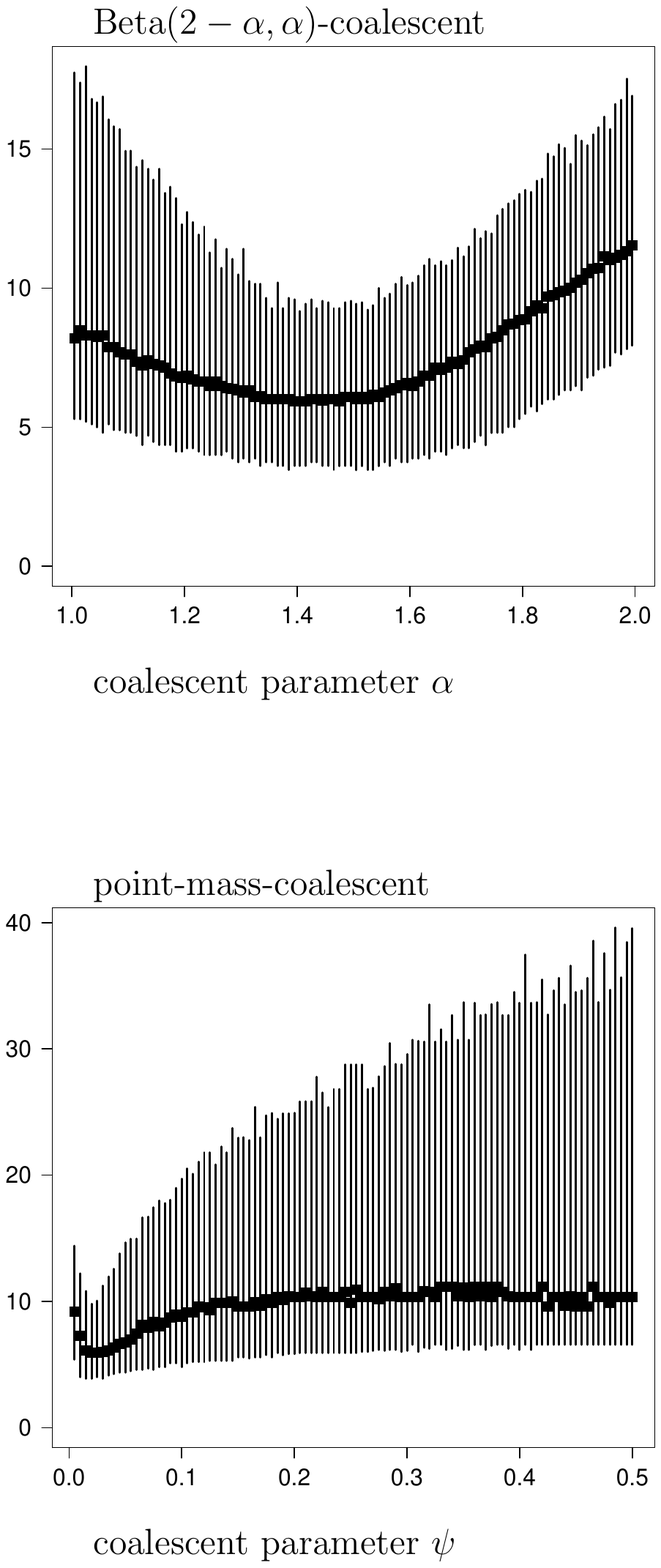}
\end{figure}%

\clearpage
\pagebreak
\newpage



\begin{figure}[H!]
  \caption{The $5\%$ and $95\%$ quantiles (lines) and the median (squares) of the $\ell^2$-distance  between simulated data for each parameter value ($\alpha$ or $\psi$) and the joint Iceland and Greenland cod dataset of \cite{A00}.  The quantiles for each parameter value are based on $1000$ iterations.   The $\alpha$ value $1.535$ yields  the minimum $95\%$ quantile of  $8.426$, while the value $0.015$ of $\psi$ yields  the minimum $95\%$ quantile of $10.863$.       }
  \label{fig:giltwodist}%
\hspace{-1in}\includegraphics[width=8in,height=8in]{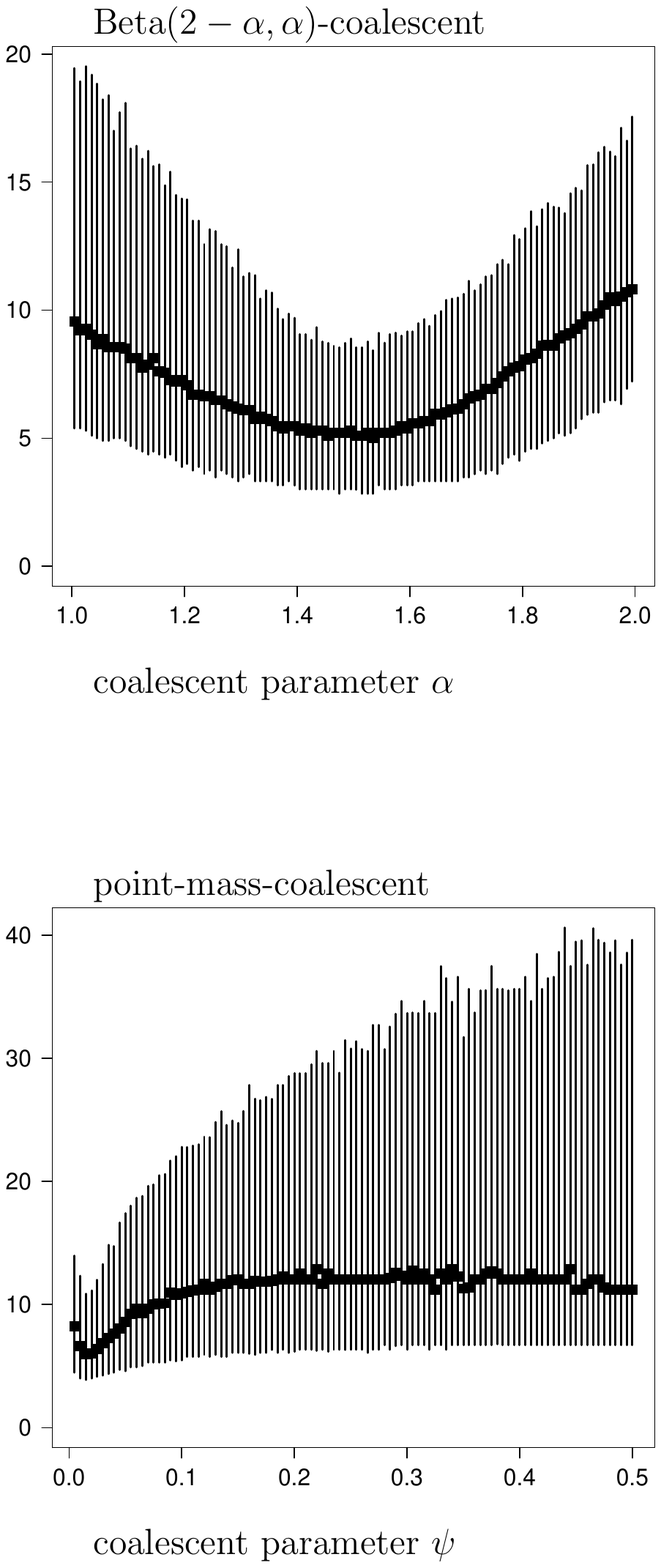}
\end{figure}%

\clearpage
\pagebreak
\newpage

\begin{figure}[H!]
  \caption{The $5\%$ and $95\%$ quantiles (lines) and the median (squares) of the $\ell^2$-distance  between simulated data for each parameter value ($\alpha$ or $\psi$) and the Atlantic cod dataset of \cite{A04}.  The quantiles for each parameter value are based on $1000$ iterations.   The $\alpha$ value $1.545$ yields  the minimum $95\%$ quantile of  $11.874$, while the value $0.01$ of $\psi$ yields  the minimum $95\%$ quantile of $16.735$.       }
  \label{fig:allltwodist}%
\hspace{-1in}\includegraphics[width=8in,height=8in]{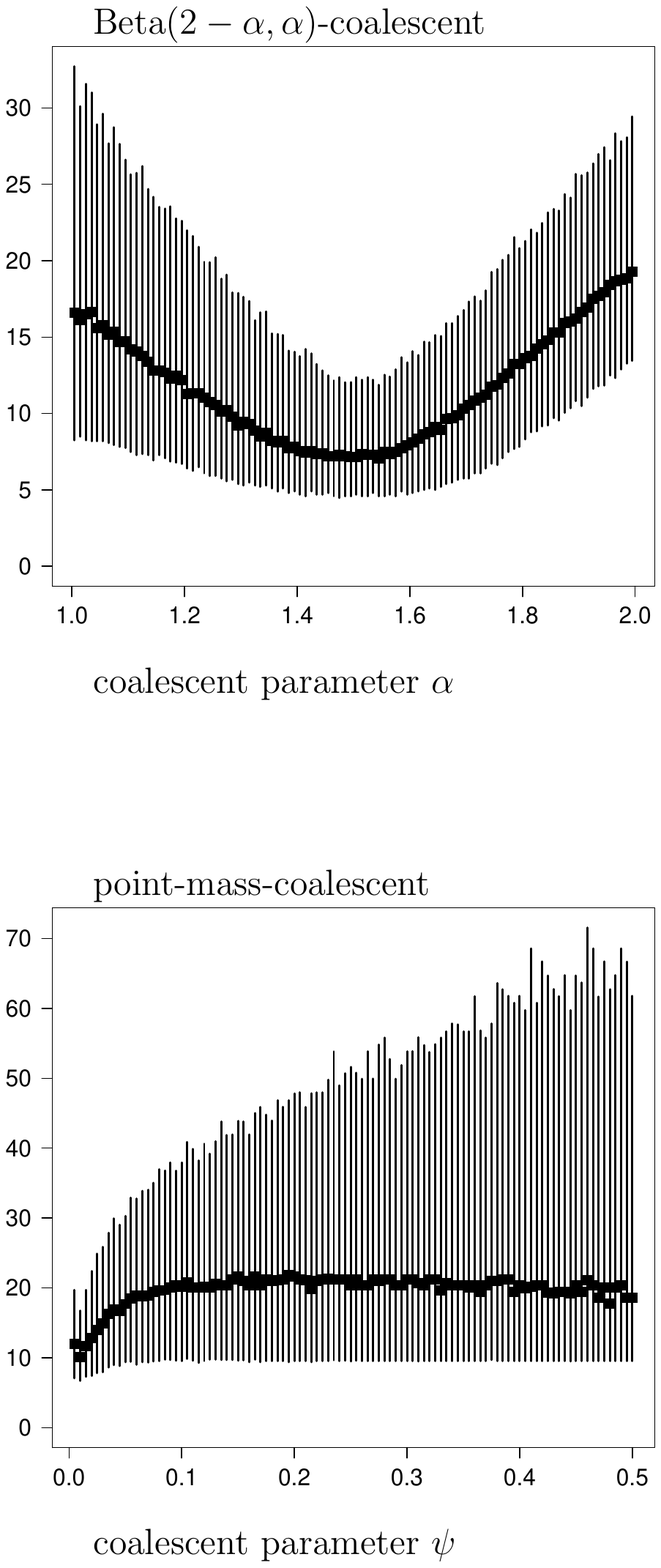}
\end{figure}%

\clearpage
\pagebreak
\newpage

\begin{figure}[H!]
  \caption{The $5\%$ and $95\%$ quantiles (lines) and the median (squares) of the $\ell^2$-distance  between simulated data for each parameter value ($\alpha$ or $\psi$) and expected values of the site-frequency spectrum associated with the Beta$(2-\alpha,\alpha)$-coalescent when  $\alpha = 1.5$, sample size is   $n = 1000$ and  $\theta = 1.0$.      The quantiles for each parameter value are based on $1000$ iterations.   The $\alpha$ value $1.505$ yields  the minimum $95\%$ quantile of  $8.035$, while the value $0.01$ of $\psi$ yields  the minimum $95\%$ quantile of $10.587$.       }
  \label{fig:betaone}%
\hspace{-1in}\includegraphics[width=8in,height=8in]{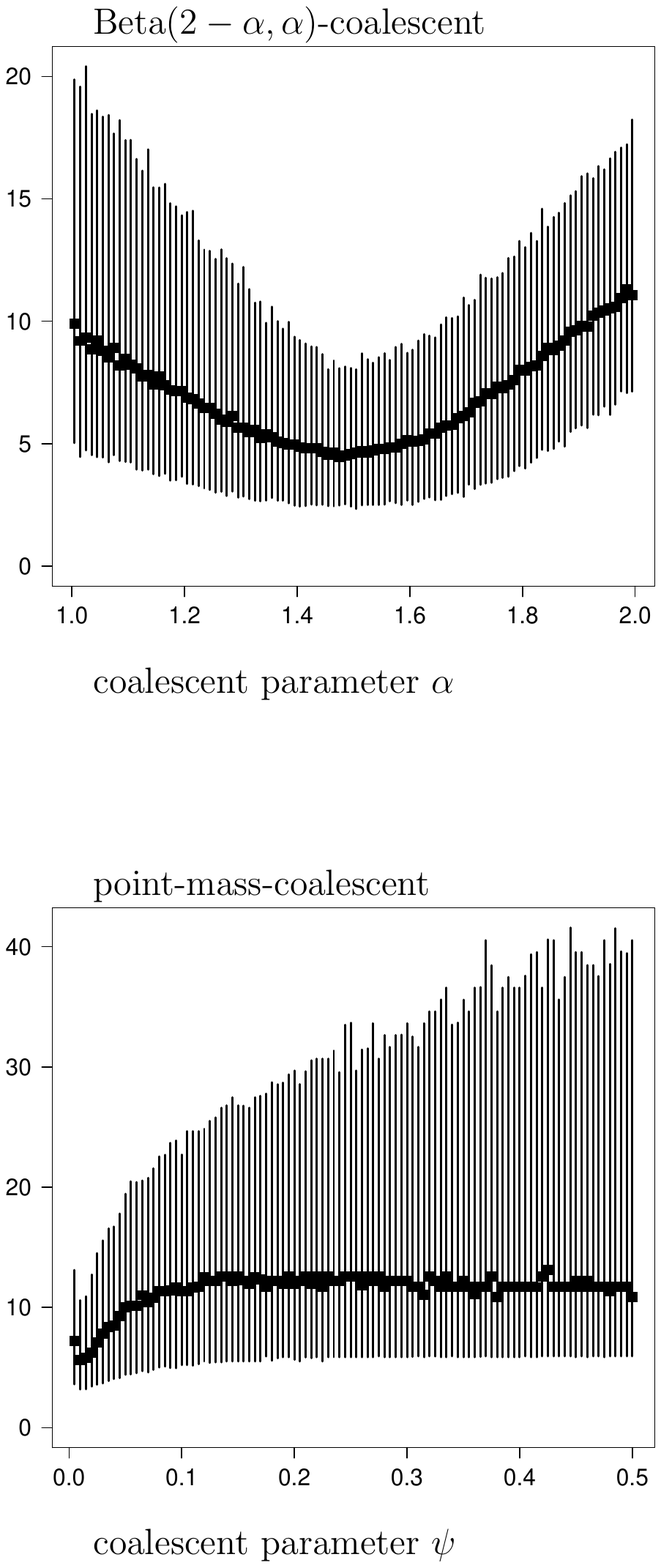}
\end{figure}%

\clearpage
\pagebreak
\newpage

\begin{figure}[H!]
  \caption{The $5\%$ and $95\%$ quantiles (lines) and the median (squares) of the $\ell^2$-distance  between simulated data for each parameter value ($\alpha$ or $\psi$) and expected values of the site-frequency spectrum associated with the Beta$(2-\alpha,\alpha)$-coalescent when  $\alpha = 1.05$, sample size is   $n = 1000$ and  $\theta = 1.0$.      The quantiles for each parameter value are based on $1000$ iterations.   The $\alpha$ value $1.205$ yields  the minimum $95\%$ quantile of  $28.472$, while the value $0.04$ of $\psi$ yields  the minimum $95\%$ quantile of $26.318$.       }
  \label{fig:betatwo}%
\hspace{-1in}\includegraphics[width=8in,height=8in]{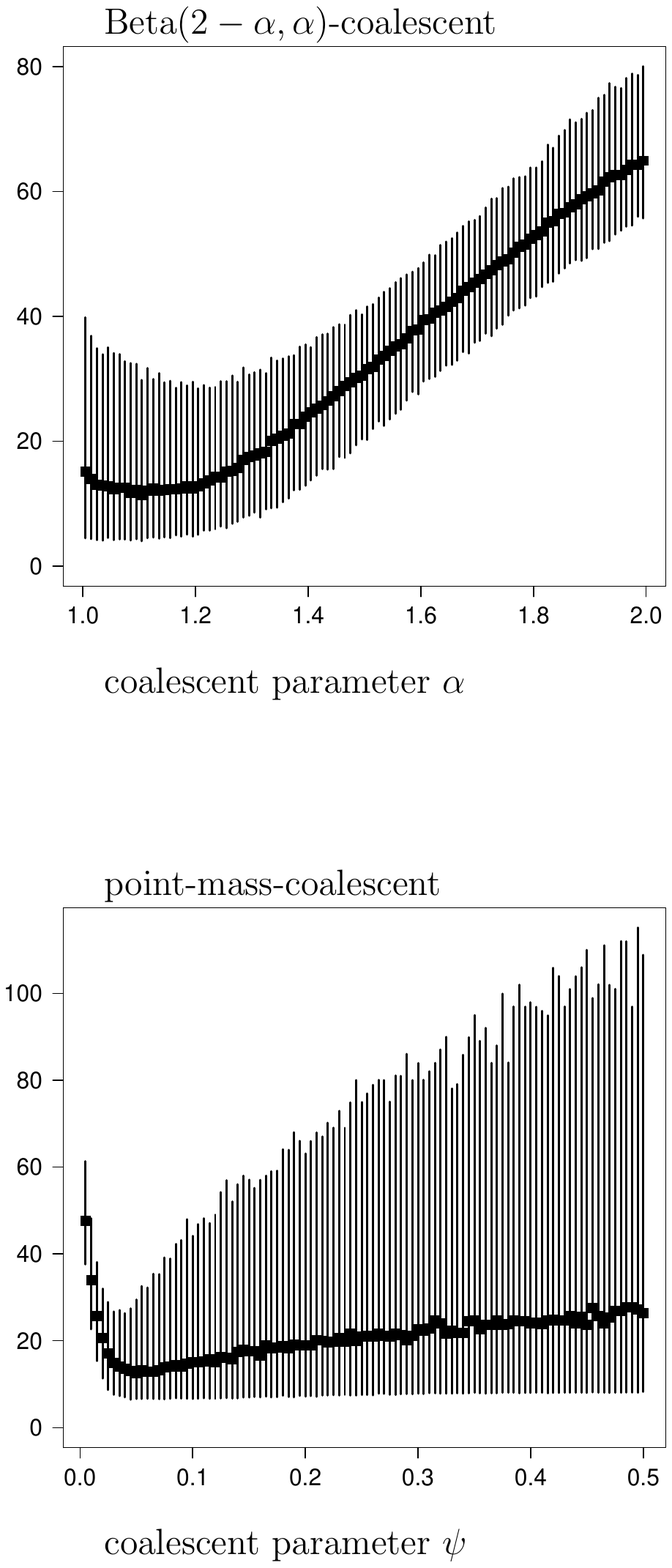}
\end{figure}%

\clearpage
\pagebreak
\newpage



\begin{figure}[H!]
  \caption{Estimates $(\hat{\alpha})$ obtained by applying the pseudo-likelihood to each of $1000$ simulated datasets for sample size $n = 100$ obtained for value of $\alpha = a$ and $\theta$ as shown.   Solid vertical lines indicate the true values of $\alpha = a$, and the broken lines indicate the $5\%$, $50\%$, and $95\%$ quantiles.   }
  \label{fig:graphlBBab}
\hspace{-1in}\includegraphics[width=8in,height=8in]{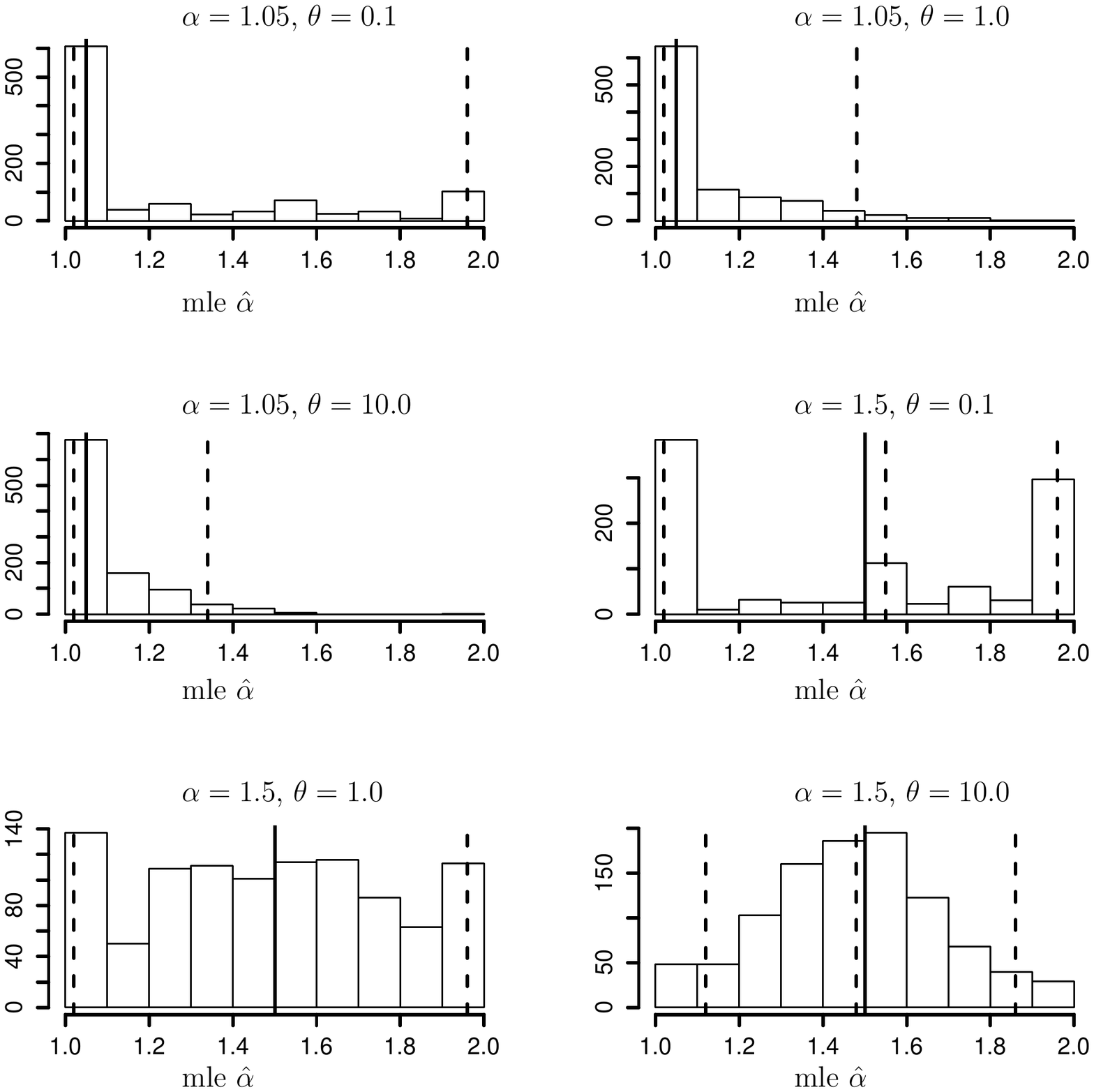}
  \end{figure}%

\clearpage
\pagebreak
\newpage

\begin{figure}[H!]
  \caption{Estimates $(\hat{\psi})$ obtained by applying the pseudo-likelihood to each of $1000$ simulated datasets for sample size $n = 100$ obtained for value of $\psi$ and $\theta$ as shown.   Solid vertical lines indicate the true values of $\psi$, and the broken lines indicate the $5\%$, $50\%$, and $95\%$ quantiles.   }
\hspace{-1in}\includegraphics[width=8in,height=8in]{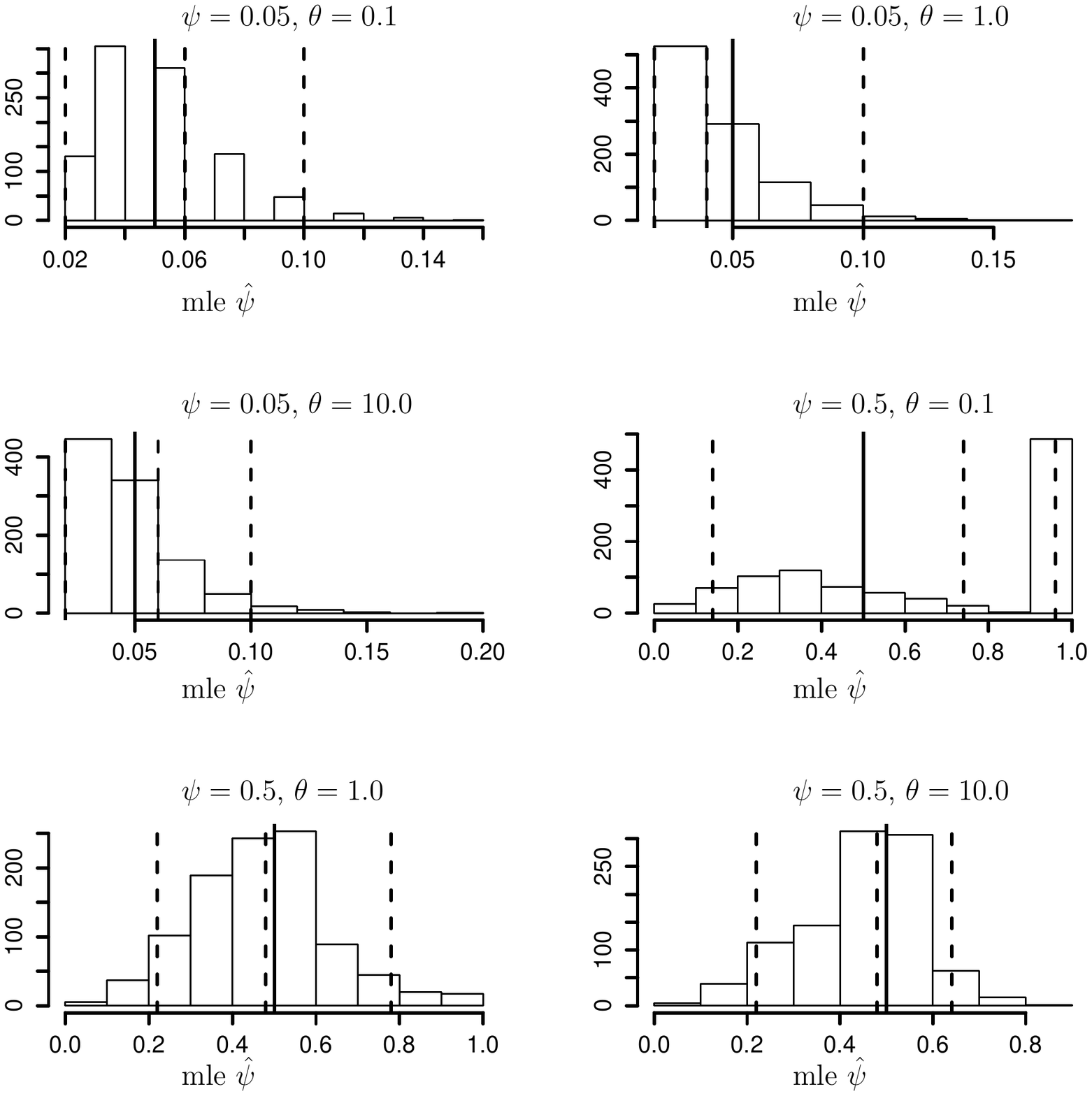}
  \end{figure}%

\clearpage
\pagebreak
\newpage

\begin{figure}[H!]
  \caption{Estimates ($\hat{\psi}$ and $\hat{\alpha}$) obtained by applying the pseudo-likelihood to each of $1000$ simulated datasets for sample size $n = 100$ under the Kingman coalescent and varying over $\theta$ as shown.   Broken lines indicate the $5\%$, $50\%$, and $95\%$ quantiles.   }
  \label{fig:mlekingman}
 \hspace{-1in}\includegraphics[width=8in,height=8in]{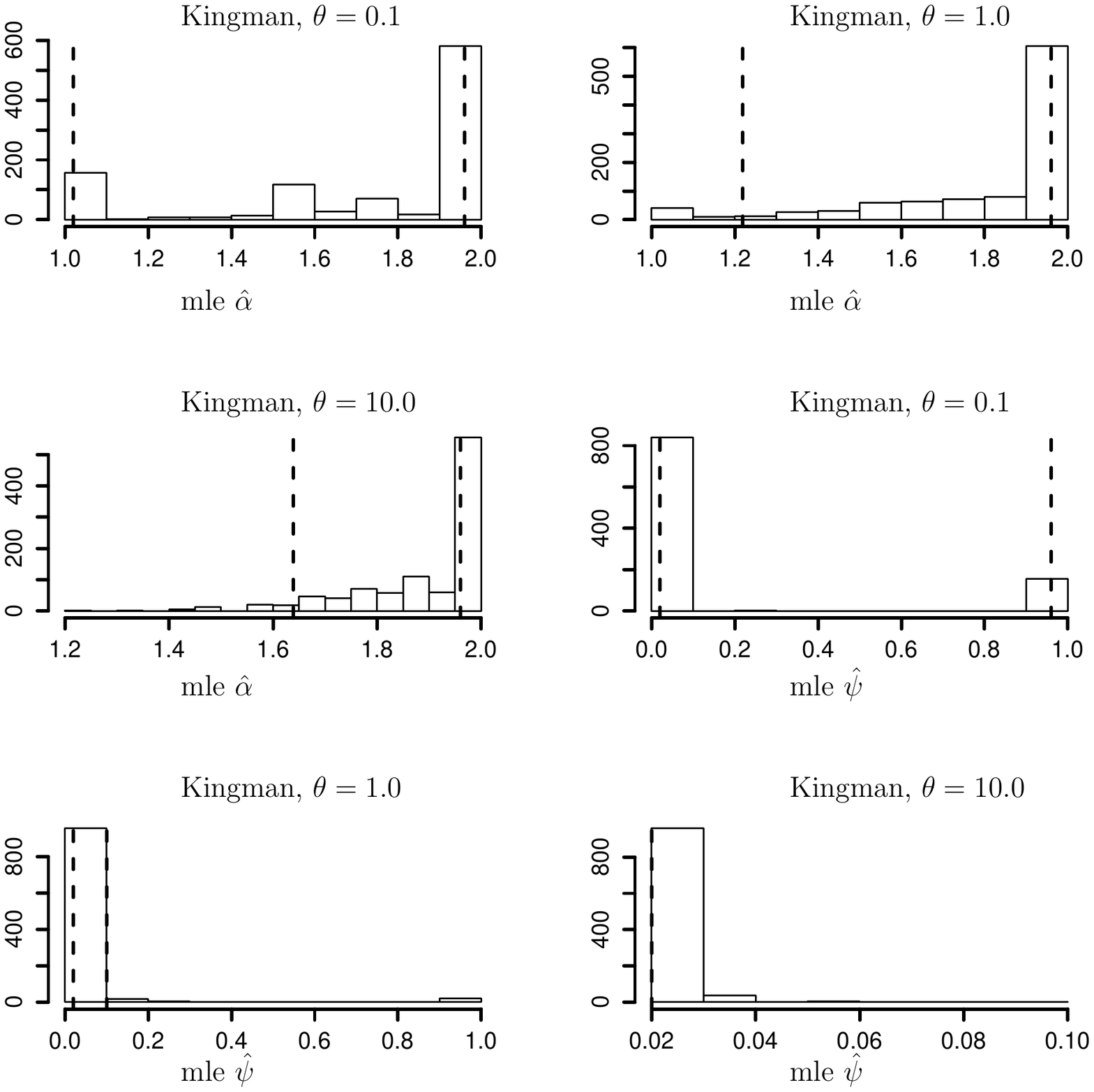}
  \end{figure}%

\end{document}